\documentclass[letterpaper,12pt,journal,onecolumn,comsoc, draftclsnofoot]{IEEEtran}
\IEEEoverridecommandlockouts

\usepackage{pifont}
\usepackage{tabu}
\usepackage{supertabular,booktabs}

\usepackage{longtable}
\usepackage[hyphens]{url}
\usepackage{multicol}
\usepackage{amsmath, amssymb, bm, cite, epsfig, psfrag, mathtools}
\usepackage{bbm}
\usepackage{epstopdf}
\usepackage{changepage}
\usepackage{graphicx}
\usepackage{fancyhdr}
\usepackage{bbm}
\usepackage{algorithm,algorithmic}
\usepackage{array}
\usepackage[margin=10pt,font=small]{caption}
\usepackage{bbm}
\usepackage{multirow}
\usepackage[usenames,dvipsnames]{xcolor}

\usepackage{subcaption}
\usepackage{etoolbox}
\usepackage{pbox}
\usepackage[width=0.48\textwidth]{caption}
\usepackage{colortbl}
\usepackage{bm}
\usepackage{caption}
\usepackage{amsmath}
\usepackage{color}
\let\labelindent\relax
\usepackage{enumitem}
\usepackage{datetime}
\usepackage{amsthm}
\usepackage{dsfont}
\usepackage{mathtools}
\usepackage[T1]{fontenc}

\newtoggle{conference}
\togglefalse{conference} 
\usepackage[letterpaper,left=0.75in,right=0.75in,top=0.75in,bottom=0.75in,footskip=.25in]{geometry}

\def\beq{\begin{equation}}
\def\eeq{\end{equation}}
\def\beqa{\begin{eqnarray}}
\def\eeqa{\end{eqnarray}}
\def\beqan{\begin{eqnarray*}}
\def\eeqan{\end{eqnarray*}}

\def\EE{{\mathbb{E}}}
\def\PP{{\mathbb{P}}}

\def\argmax{\mathop{\mathrm{arg\,max}}}

\newcommand{\Ic}{{\cal I}}
\newcommand{\Jc}{{\cal J}}

\newtheorem{remark}{Remark}

\newtheorem{definition}{Definition}
\newtheorem{theorem}{Theorem}
\newtheorem{lemma}{Lemma}
\newtheorem{example}{Example}

\def\tm1{t\! - \! 1}
\def\tp1{t\! + \! 1}

\def\dbf{\mathbf{d}}

\def\Cbf{\mathbf{C}}

\newcommand{\Ec}{{\cal E}}

\def\Qbf{\mathbf{Q}}

\newcommand{\Vc}{{\cal V}}
 \def\Gc{\mathcal{G}}

\def\Wc{\mathcal{W}}
\def\Yc{\mathcal{Y}}

\newcommand{\Kc}{{\mathcal K}}

\def\Ysf{ {\sf Y}}

\def\Wsf{ { W}}

\def\Rach{R_{A}} 
\def\File{ {U}}

\begin{document}

\bibliographystyle{IEEEtran}


\title{On Coding for Cache-Aided Delivery of Dynamic Correlated Content}

\author{Parisa Hassanzadeh, Antonia M. Tulino, Jaime Llorca, Elza Erkip
\thanks{This work has been supported in part by NSF under Grant \#1619129 and in part by NYU WIRELESS.}
\thanks{P. Hassanzadeh  and  E. Erkip are with the ECE Department of New York University, Brooklyn, NY. Email: \{ph990, elza\}@nyu.edu}
\thanks{J. Llorca  and A. Tulino are with Nokia Bell Labs, Holmdel, NJ, USA. Email:  \{jaime.llorca, a.tulino\}@nokia-bell-labs.com}
\thanks{A. Tulino is with the DIETI, University of Naples Federico II, Italy. Email:  \{antoniamaria.tulino\}@unina.it}
}

\maketitle

\begin{abstract}
Cache-aided coded multicast leverages side information at wireless edge caches to efficiently serve multiple unicast demands
via common multicast transmissions, leading to load reductions that are proportional to the aggregate cache size. 
However, the increasingly dynamic, unpredictable, and personalized nature of the content that users consume challenges the efficiency of existing caching-based solutions in which only {\em exact} content reuse is explored. 
This paper generalizes the cache-aided coded multicast problem  to specifically account for the correlation among content files, such as, for example, the one between updated versions of dynamic data. 
It is shown that (i) caching content pieces based on their correlation with the rest of the library, and (ii) jointly compressing requested files using cached information as references during delivery, can provide load reductions that go beyond those achieved with existing schemes. This is accomplished via the design of a class of correlation-aware achievable schemes, shown to significantly outperform state-of-the-art correlation-unaware solutions. Our results show that as we move towards real-time and/or personalized media dominated services, where exact cache hits are almost non-existent but updates can exhibit high levels of correlation, network cached information can still be useful as references for network compression. 
\end{abstract}


\section{Introduction}~\label{sec:Introduction}

Proper distribution of popular content across wireless edge caches is emerging as a promising approach to address the exponentially growing traffic in current wireless networks. 
Recent studies have shown that, in a cache-aided network, exploiting globally cached information in order to multicast coded messages that are useful to a large number of receivers exhibits overall network load reductions that are proportional to the aggregate cache capacity. The fundamental rate-memory trade-off for a broadcast caching network in a variety of settings that includes worst-case demands, random demands, multiple per-user requests, heterogeneous channel conditions, and finite packetization, has been characterized in \cite{maddah14fundamental,maddah14decentralized,ji2017order,ji14average,ji14groupcast,ji15multiple,channel2016,shanmugam14finite,yu2016exact}. 
While these results are promising, most existing studies treat the network content as independent pieces of information, and do not account for the additional gains that can be obtained from the joint compression of correlated content distributed throughout the network. 
Exploiting content correlation becomes particularly critical as we move from static content distribution towards real-time delivery of rapidly changing (and aging), but highly correlated, personalized data (as in news updates, social networks, remote sensing, augmented reality, etc.) in which {\em exact} content reuse is almost non-existent \cite{iot2012, Antony17aoi}.

In this paper, we investigate how content correlations can be explored in order to improve the performance of cache-aided networks. 
We consider a network setup similar to \cite{maddah14fundamental,maddah14decentralized,ji2017order,ji14average,ji14groupcast,ji15multiple,channel2016,shanmugam14finite,yu2016exact}, but assume that the files in the library are generated according a joint distribution and are therefore correlated. 
Such correlations are especially relevant among content files of the same category, such as episodes of a TV show or same-sport recordings, which, even if personalized, may share common backgrounds and scene objects. 
In addition, such correlations are also fundamental in modeling dynamic data, where information exhibits rapid {\em aging} (i.e., loses relevance), and requires timely updates. This calls for new strategies that are able to efficiently exploit the correlation between multiple versions of dynamic data in order to enable both 
compressed delivery, as well as efficient updating of cached references.  

As in existing literature on cache-aided networks, we assume that the network operates in two phases: a caching (or placement) phase taking place at network setup followed by a delivery phase where the network is used repeatedly in order to satisfy receiver demands. The design of the caching and delivery phases forms what is referred to as a {\em caching scheme}. During the caching phase, caches are filled with content from the library according to a properly designed {\em caching distribution}. During the delivery phase, the sender compresses the set of requested files into a multicast codeword by computing an {\em index code} \cite{birk1998informed,IndexCoding}. The goal of this paper is to investigate the additional gains that can be obtained when accounting for the correlations that may exist among static files as well as multiple versions of dynamic data, when designing both caching and delivery phases.

\subsection{Related Work}
There are only a few works that study the cache-aided broadcast network for a library composed of correlated files \cite{timo2018rate,hassanzadeh2017rate,Asilomar2017,ISITjournal,ITW2016,yang2017centralized,ISTC2016}. 
The works in \cite{timo2018rate,hassanzadeh2017rate,Asilomar2017,ISITjournal} analyze the rate-memory trade-off in the correlated setting via information theoretic approaches. Timo et. al.
\cite{timo2018rate} provide rate-memory-distortion trade-offs for a lossy reconstruction setting with two receivers, multiple files, and a single cache, for which local caching gains are exploited. The authors determine the information common to all sources to be the most useful content to be placed in the cache. In order to capitalize on the global caching gains arising from multiple caches in the network, the rate-memory trade-off under lossless reconstruction for the two-receiver, two-cache, two-file network was studied in \cite{hassanzadeh2017rate}, and extended to multiple files and multiple receivers in \cite{Asilomar2017,ISITjournal}. 
In these papers, an achievable two-step scheme is proposed to exploit the correlation by jointly compressing the library files prior to the caching phase based on the Gray-Wyner network \cite{gray1974source}, and then treating the compressed content as independent files. It is shown that this strategy is optimal for a large memory regime, while the gap to optimality is quantified for other memory values. However, the exponential complexity of Gray-Wyner source coding  makes the overall characterization with large number of files difficult. The work in \cite{ITW2016} addresses this complexity by proposing an achievable scheme that takes on a more practical approach to finding the information common to a set of files. To this end, the library files are grouped into sets of correlated files, and each file is compressed with respect to the file in its group that leads to the highest compression.  In \cite{yang2017centralized}, the authors propose a caching scheme for arbitrary number of files and receivers, where the library has a special correlation structure, i.e., each library file is composed of  multiple  independent subfiles  that are common among a fix set of files in the library.  They numerically compare the performance of their scheme with a lower bound on the worst-case rate-memory trade-off.

Even though jointly compressing the library before the caching phase, as in \cite{timo2018rate,hassanzadeh2017rate,Asilomar2017,ISITjournal,ITW2016,yang2017centralized}, is a natural solution to reducing the load on the shared link, this approach is mainly suitable for a library with static content, losing robustness and leading to significant performance degradation when content is highly dynamic. For example, if the library is updated and a new file is added, all files need to be jointly re-compressed and the content cached across the entire network needs to be updated based on the new compressed versions. On the other hand, schemes that are not based on a-priori library compression may only require updating a portion of the network cached content, which is  especially critical for the delivery of dynamic data in next generation services.

\subsection{Contributions}
Motivated by the increasing dynamic and personalized nature of next generation content services, in this paper we consider the problem of efficient caching and delivery of correlated content in a general setting with multiple receivers and files and arbitrary joint file distribution. Our goal is to design a practical scheme that is robust to the variations inherent to the delivery of real-time services over wireless networks.  To this end, we propose an achievable scheme that stores individually (rather than jointly) compressed  content pieces during the caching phase, which then serve as references for compression during the delivery phase,   resulting in an efficient scheme that is robust to dynamic changes in the content library. Compared with previously proposed schemes \cite{timo2018rate,hassanzadeh2017rate,Asilomar2017,ISITjournal,ITW2016,yang2017centralized}, where correlation-aware compression takes place before the caching phase, our scheme provides on-demand compression during the delivery phase and hence is able to adapt to system variations while continuing to exploit available content correlations. This paper provides a comprehensive analysis of the proposed scheme, for which preliminary results were presented in \cite{ISTC2016}.

Our main contributions are summarized as follows:
\begin{itemize}
	
	\item We formulate the problem of efficient delivery of {\em dynamic} and {\em correlated} sources over a broadcast caching network via information-theoretic tools.  
	
	\item We propose a correlation-aware scheme that consists of (i) storing individually compressed content pieces based on their popularity as well as on their correlation with the rest of the library in the caching phase, and (ii) sending compressed versions of the requested files according to the information distributed throughout the network and their joint statistics during the delivery phase.   By using an individually compressed caching phase, and providing {\em on-demand} compression during the delivery phase, our scheme is robust to changes in the content library, particularly relevant in next generation dynamic content services.  	
	
	\item We analyze the performance of our proposed correlation-aware scheme for two settings: $i)$ a {\em static} setting, in which we assume that the same (correlated) content library is used during both the caching and delivery phases,  and $ii)$ a {\em dynamic} setting, in which  an updated version of the content library may become available during the delivery phase.  
	
	\item We characterize an upper bound on the expected rate of the proposed correlation-aware scheme in a  network with arbitrary number of files and receivers, and numerically compare the achievable rate 	with that of existing correlation-unaware schemes, which confirm the additional available gains, especially for small memory sizes. 
	In addition, for the special case of the two-file two-receiver network considered in \cite{hassanzadeh2017rate,ISITjournal}, we present a scheme based on a more structured cache placement that achieves a rate within half of
	the mutual information of the two files for all cache sizes.
\end{itemize}

\subsection{Paper Organization}
The paper is organized as follows. 
Sec. \ref{sec:Problem Formulation} presents the information-theoretic problem formulation. Sec.~\ref{sec:ach scheme} provides a detailed description of the proposed achievable scheme. We upper bound the expected rate-memory trade-off of the proposed scheme with randomized and deterministic cache placements in Secs.~\ref{Sec: upper bound} and \ref{Sec: 2user2file}, respectively. The performance of the scheme is numerically validated in Sec. \ref{sec:Simulations}, and the conclusion follows in Sec. \ref{sec:Conclusion}.

\section{Network Model and Problem Formulation}\label{sec:Problem Formulation}
We consider a broadcast caching network composed of one sender (e.g., base station) with access to a library composed of $N$ files, generated by an $N$-component discrete memoryless source (N-DMS). 
The files, denoted by $\{W_1^F,\dots, W_N^F\}$, are of length $F$, with file $n \in \{1,\dots,N\}$ denoted by  
$\Wsf_n^F = [\Wsf_{n}(1),\dots,\Wsf_{n}(F)]$. 
For $i\in\{1,\dots F\}$, the $N$-dimensional vectors  $[\Wsf_{1}(i),\dots,\Wsf_{N}(i)]$, which represent the  $i^\text{th}$ elements of the files, are independently and identically distributed (i.i.d.) according to $[W_1,\dots,W_N]\sim p({\bf w})$, with the probability mass function (pmf) $p({\bf w})=p (w_1, \dots, w_N)$ defined over an alphabet $\Wc^N$. 
%
Without loss of generality, we assume that $H(W_n)=H(W)$ for $n\in\{1,\dots,N\}$. 
The sender communicates with $K$ receivers (e.g., access points or user devices) $\{1,\dots,K\}$, through a shared error-free multicast link. Each receiver has a cache of size $MF$ bits, where $M\in[0,\,NH(W)]$ denotes the (normalized) cache capacity. The network operates in two phases:
\begin{itemize}
	\item[$i)$] A caching phase that takes place at the network setup, 
	and is assumed to happen during off-peak hours without consuming actual delivery rate.
	During this phase,  the receivers fill their caches from $\{W_1^F,\dots, W_N^F\}$, referred to as the {\em original library}. 
	
	\item[$ii)$] {A delivery phase where the network is repeatedly used,  and library files can be dynamically updated.  
		The most recent versions of the files are denoted by $\{V_1^F,\dots, V_N^F\}$. For each $i\in\{1,\dots,F\}$, conditioned on $[\Wsf_{1}(i),\dots,\Wsf_{N}(i)]$, the vectors  $[V_{1}(i),\dots,V_{N}(i)]$ are  i.i.d. 
		according to $ [V_1,\dots, V_N]\sim p({\bf v}|{\bf w})=p(v_1,\dots,v_N|{\bf w})$. 
		Without loss of generality, we assume the most recent versions have the same alphabets as the original files.  Availability of the most recent content at the server is subject to network delays, and as a result, we assume that only with probability $\pi_n$, the server observes the most recent version of file $n\in\{1, \ldots, N\}$. Hence the {\em updated library}, denoted by $\{\File_1^F,\dots, \File_N^F\}$, is such that $\File_n^F=V_n^F$ with probability $\pi_n$ and $\File_n^F=W_n^F$ with probability $1-\pi_n$, $n\in\{1, \ldots, N\}$, where availability of the most recent content at the server is modeled as an independent process across files.   
		We also assume that in different uses of the network, the most recent versions are obtained conditionally i.i.d. with respect to the original library, and files are updated  at the server (with respect to the cached versions) in an i.i.d fashion. The {\em library distribution} is then characterized by $p({\bf w},{\bf v})$ and ${\bm \pi}=(\pi_1, \ldots, \pi_N)$.} Receivers request files in an i.i.d. manner according to a uniform demand distribution, and the sender satisfies the demands using the updated library. The demand realization is denoted by $\dbf=(d_1,\dots,d_K)$, where $\dbf\in \mathcal D \equiv \{1,\dots,N\}^K$. 
\end{itemize}
A caching scheme for this network consists of:
\begin{itemize}
	\item {\textbf{Cache Encoder:}} The cache encoder at the sender computes the content to be cached at receiver $k \in \{1,\dots,K\}$, denoted by $Z_{k}$, using the set of functions   
	$f^{\mathfrak  C}_{k}:   {\mathcal W}^{NF} \rightarrow [1: {  2^{MF}}  ),$
	as
	$$Z_{k} = f^{\mathfrak  C}_{k}\Big(  W_1^F,\dots,W_N^F  \Big).$$
	The encoder designs the cache configuration $\{Z_{1},\dots,Z_{K}\}$ jointly across receivers, taking into account available global system knowledge such as, the number of receivers and their cache sizes, the number of files, and the pmf $p({\bf w})$.  
	
	\item{\textbf{Multicast Encoder:}} 
	After the caches are populated, the network is repeatedly used. At each use, a demand $\dbf$ from the updated library, $\{\File_1^F,\dots,\File_N^F\}$, is revealed to the sender. The multicast encoder uses a fixed-to-variable encoding function
	$$f^{\mathfrak  M}: {\mathcal D}   \times  {\mathcal W}^{NF} \times [1: 2^{MF})^K \rightarrow \Yc^\star,$$
	to generate a codeword,  $Y_{\dbf}$, according to the demand realization $\dbf$, the updated library $\{\File_1^F,\dots,\File_N^F\}$, the cache configuration $\{Z_{1},\dots,Z_{K}\}$, and the library distribution. Then,
	$$Y_{\dbf} = f^{\mathfrak  M}\Big( \dbf ,\, \{ {\File_n^F\}_{n=1}^N},\,   \{Z_{k}\}_{k=1}^K \Big)  ,$$ 
	where $\Yc^\star$ denotes the set of finite length sequences.
	
	\item{\textbf{Multicast Decoders:}} Receivers recover their requested files using their cached content and the received
	multicast codeword. Receiver ${k}$ recovers $\File^F_{d_{k}}$ using decoding function $h^{\mathfrak  M}_{k} :
	\mathcal D \times \Yc^\star \times [1: 2^{MF}) \rightarrow \Wc^F $, as
	$${\widehat{\File}^F_{d_{k}} }= h^{\mathfrak  M}_{k}(\dbf, Y_{\dbf} ,Z_{k})    .$$
\end{itemize}
We refer to the overall scheme as a {\em cache-aided coded multicast (CACM)} scheme. The worst-case probability of error of the corresponding CACM scheme is defined as
\begin{align} \label{perr}
& P_e^{(F)} = \max_{\dbf \in \mathcal D}\,   \PP \left( \bigcup\limits_{  k\in\{1,\dots,K\}}\, \Big\{\widehat{\File}^F_{d_{k}}  \neq \File^F_{d_{k}}  \Big\} \right), \notag
\end{align}
which is averaged over the  library distribution. 
In line with previous work \cite{ji2017order,ji14average,ji14groupcast},  the  average multicast delivery rate  (or load) of the overall scheme, over all demands, is defined as
\begin{equation} \label{average-rate}
R^{(F)} =    \frac{\EE[L(Y_{\dbf})]}{F}.
\end{equation}
where $L(Y)$ denotes the length (in bits) of the multicast codeword $Y$, and the expectation is over the demands and the library files used during caching and delivery phases.

\begin{definition} \label{def:achievable-rate}
	A rate-memory pair $(R,M)$ is {\em achievable} if there exists a sequence of caching schemes for memory $M$ and
	increasing file size $F$ such that {$\lim_{F \rightarrow \infty} P_e^{(F)} = 0 \notag$, and $\limsup_{F \rightarrow \infty}
		R^{(F)} \leq  R.$}
\end{definition}
\begin{definition} \label{def:infimum-rate}
	The rate-memory region is the closure of the set of achievable rate-memory pairs $(R,M)$. The optimal rate-memory function, $R^*(M)$, is the
	infimum of all rates $R$ such that $(R,M)$ is in the rate-memory region for memory $M$.
\end{definition}

The goal of this paper is to design CACM schemes that result in a small achievable rate $R$ for a given memory $M$. We propose a class of CACM schemes in Sec.~\ref{sec:ach scheme} for the general system model provided here. Later in the paper, we analyze the performance of the proposed scheme by considering the following two settings:
\begin{itemize}
	\item[-] {\em \bf Static Setting}: In this setting, the library remains static throughout the network setup and its consecutive uses, such that content is requested and delivered from the same library that is used in the caching phase; that is, $\pi_n=0$, $n\in\{1, \ldots, N\}$.  Hence, in this scenario the library is composed of static correlated content, as in \cite{hassanzadeh2017rate,Asilomar2017,ISITjournal,timo2018rate,ITW2016,yang2017centralized}, and when particularized to independent files, the setting is equivalent to that of most prior works, including
	\cite{maddah14fundamental,maddah14decentralized,ji2017order}, and \cite{ji14groupcast,ji15multiple,channel2016,shanmugam14finite}.

	\item[-] {\em \bf Dynamic Setting}: In this case, for ease of exposition, we assume that the library used during the caching phase is originally composed of $N$ independent files, i.e.,  $p({\bf w}) = \prod\limits_{n=1}^N p(w_n)$, and the most recent versions of the files are generated such that $p({\bf v}|{\bf w}) = \prod\limits_{n=1}^N p(v_n|w_n)$, but otherwise the library distribution follows the general model described in this section.

\end{itemize}

\section{Proposed Correlation-Aware Scheme}\label{sec:ach scheme}
In this section, we introduce a class of correlation-aware schemes, which distribute the most relevant library content among the receiver caches such that during the delivery phase, the correlation among the aggregate cache and the demand can be exploited to send multicast codewords of compressed files that  further improve the global caching gain. In the following, we refer to the overall scheme as {\em correlation-aware CACM} {(CA-CACM)}. Before providing the general description of the proposed CA-CACM scheme, we first illustrate the main idea of how correlation is exploited during the delivery phase through a simple example in a dynamic setting.

\vspace{2mm}
{\bf Motivating Example:}

Consider a network with two receivers, two independent files $\{W_1^F ,W_2^F \}$, cache capacity $M=1$, and cached content
\begin{align}
Z_1 =\{W_{1,1},\,W_{2,1} \},\quad Z_2 =\{ W_{1,2},\,W_{2,2} \},\label{eq:cache dynamic}
\end{align}	
where $W_{i,j}$ denotes the $j^\text{th}$ half of file $W_i^F$ with length $F/2$ bits. Cache configuration \eqref{eq:cache dynamic}, in which each receiver stores an exclusive part of each file, is optimal for the static system with independent files considered in \cite{maddah14fundamental}. With such caching, a single coded transmission over the shared link helps both receivers to effectively exchange the missing packet available in the cache of the other receiver.  
During the delivery phase,  with probability $\pi_1=\pi_2=1$, new versions of the files become available, $\{\File_1^F ,\File_2^F \}$, which are highly correlated with the versions used during the caching phase, such that $H(\File_1|W_1)=H(\File_2|W_2 ) = 0.5$. As a worst-case demand, assume user 1 requests file $\File_1^F$ and user 2 requests file $\File_2^F$. Since none of the requested files are cached, there are no exact matches between the cached content and the requested files, and hence, correlation-unaware schemes can not utilize the cached content. Similarly, the correlation-aware scheme proposed in \cite{ISITjournal} falls short since the original library used for caching was composed of independent files.  For either of the schemes, the requested files must be delivered via separate uncoded transmissions $\File_1^F$ and $\File_2^F$, yielding a delivery load equal to $R=2$. However, {\em aware} of the correlation among the requested files and the cached content, the  sender can transmit the codeword  $W_{1,2}\oplus W_{2,1}$, concatenated with two refinements with rates $H(\File_1|W_1 )$ and $H(\File_2|W_2 )$. Then: i) receiver 1 is able to losslessly recover $\File_{1}^F$ from its decoded message $W_{1,2}$ and stored message $W_{1,1}$, and ii) receiver 2 is able to recover $\File_{2}^F$ from its decoded message $W_{2,2}$ and stored message $W_{2,1}$. This results in a total load of $R= 0.5+ H(\File_1|W_1) + H(\File_2|W_2 )=1.5$.

From the example, it is observed that placing {\em individually} compressed content pieces during the caching phase, which serve as references for compression during the delivery phase, enables {\em on-demand} compression suitable for dynamic correlated content. The next section describes the scheme, which is general enough to exploit correlation in static settings, as the one in Sec.~\ref{Sec: 2user2file}, and to provide robustness in dynamic settings such as the one given above.

\vspace{1mm}
{\bf Proposed Scheme:}

A CA-CACM scheme, as depicted in Fig.~\ref{fig:encoderblock2}, is composed of the following components (described in more detail in the following subsections): 

\begin{itemize}
	\item[$i)$]{\em Correlation-Aware Cache Encoder}: As in conventional CACM schemes, a fractional cache encoder is used to divide each file into packets and determine the subset of each file's packets to be cached at each receiver. This can be done in a {\em deterministic} fashion, referred to as {centralized caching} in the literature, as in \cite{maddah14fundamental,ji14groupcast}, and \cite{yu2016exact}, or, in order to be more robust to system dynamics, in a {\em random} fashion, referred to as decentralized caching, as in \cite{maddah14decentralized,ji2017order}, and \cite{ji15multiple,channel2016,shanmugam14finite}, where the packets to be cached are selected according to a {\em caching distribution}. 
	The caching distribution is optimized in order to minimize the achievable delivery rate with respect to the library distribution, as a function of not only file popularity, but also their cross-correlation. 		
	\item[$ii)$] {\em Correlation-Aware Multicast Encoder:}
	\begin{itemize}
		\item[$\bullet$]{\em Correlation-Aware Grouping}: Given a demand realization in the delivery phase, the multicast encoder identifies all the packets among the collective cache and demand that are correlated with each requested packet that is not locally available. Any of the identified correlated packets can be used in the multicast codeword in place of the requested packet, subject to transmitting additional codewords that ensure lossless reconstruction of the requested packet. 
		\item[$\bullet$]{\em Correlation-Aware Group Coloring}: The multicast encoder generates a multicast codeword constructed as a linear combination of the requested packets and their correlated packets concatenated with all necessary refinements, such that all receivers can losslessly recover their demands. This is done by 
		computing a {\em group coloring} of a graph that results from augmenting the index coding conflict graph \cite{birk1998informed,IndexCoding} with the correlated ensembles computed by the correlation-aware grouping procedure. In the following, we refer to this graph as the {\em augmented conflict graph}.
	\end{itemize}	
\end{itemize}

\begin{figure}
	\centering
	\includegraphics[width= 0.72\linewidth]{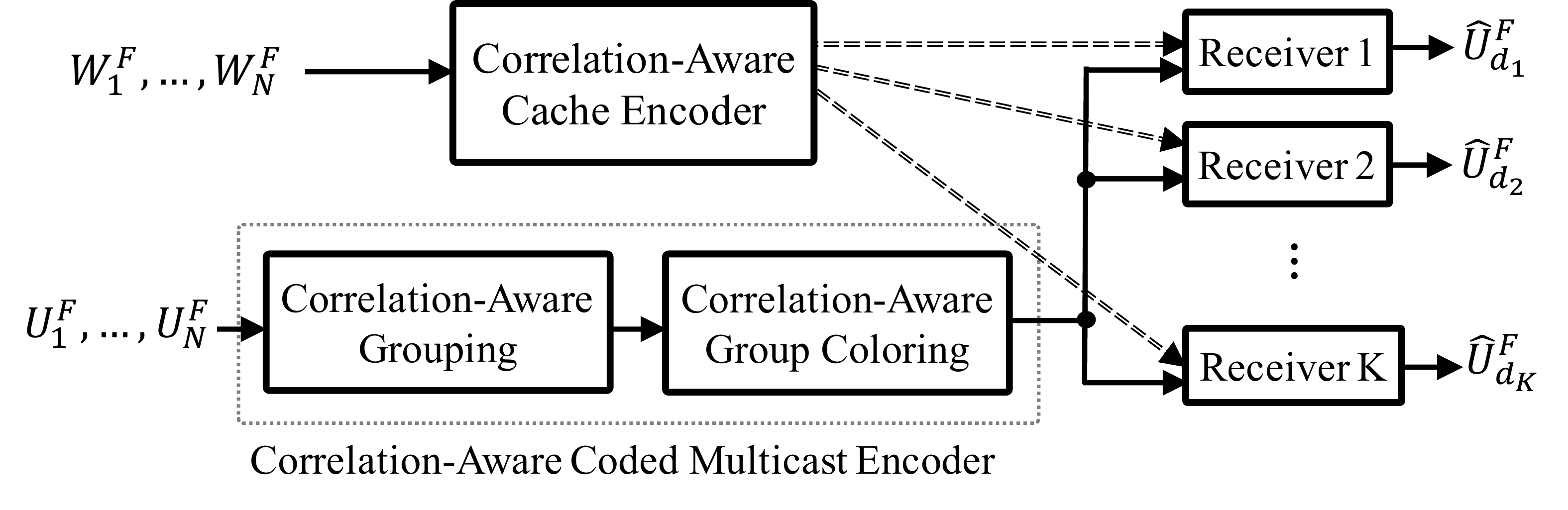}
	\caption{The proposed correlation-aware CACM scheme.}
	\label{fig:encoderblock2}
\end{figure}

\vspace{2mm}
\subsection{Correlation-Aware Cache Encoder}\label{subsec:CA RAP} 
As stated earlier, the correlation-aware cache encoder can be designed in a deterministic or random fashion.  Differently from the random cache placement, where only  the portion of packets to be cached from each file is stated, deterministic cache placement also specifies their identity, and therefore, results in better performance  compared to random caching. 
For a detailed description of the deterministic caching we refer the reader to Sec.~\ref{Sec: 2user2file}, and in this section, we describe the random version of the proposed cache encoder, which we refer to as the random fractional caching strategy. In addition to providing increased robustness to system dynamics, this strategy also allows us to analytically derive an upper bound on the  rate achieved by the proposed CA-CACM scheme in Sec.~\ref{Sec: upper bound}.

Each file $W_{n}^F$ is partitioned into $B$ equal-size packets\footnote{For example, in video applications, a packet may represent a video block or frame.}, where packet $b\in\{1,\dots,B\}$ of file $n \in \{1,\dots,N\}$ is denoted by $W_{n,b}$. The content to be cached at each receiver is determined according to a \textit{caching distribution}, ${\bm \varrho} = (\varrho_{1}, \dots, \varrho_{N})$ with $0 \leq \varrho_{n} \leq 1/M$, $\forall n \in\{1,\dots,N\}$ and $\sum_{n = 1}^{N} \varrho_{n} = 1$, which is optimized to minimize the rate of the corresponding index coding delivery scheme. For a given caching distribution $\bm \varrho$, each receiver randomly selects and caches a subset of $\varrho_{n}MB$ distinct packets from file $n \in \{1,\dots,N\}$. We denote by $\Cbf=\{ \Cbf_1,\dots,\Cbf_K\}$ the packet-level cache configuration, where $\Cbf_k$ denotes the set of packets $W_{n,b}$, $n\in\{1,\dots,N\}$, $b\in\{1,\dots,B\}$, cached at receiver $k$. 
While the caching distribution of a correlation-unaware scheme prioritizes the caching of packets according to the aggregate popularity distribution (see \cite{ji14average, ji2017order}), the correlation-aware caching distribution accounts for both the aggregate popularity and the correlation among the library files when determining the amount of packets to be cached from each file.

\subsection{Correlation-Aware Multicast Encoder}\label{subsec:CA Delivery} 

The correlation-aware coded multicast encoder capitalizes on the additional coded multicast opportunities that arise from incorporating cached packets that are, not only equal to, but also correlated with the requested packets into the multicast codeword. For a given demand realization $\dbf$, the packet-level demand realization is denoted by $\Qbf=\{ \Qbf_1,\dots,\Qbf_K \}$, where $\Qbf_{k}$ denotes the packets of file $\File_{d_k}^F$ requested, but not cached,  by receiver $k$. 

{\bf Correlation-Aware Grouping.} For each requested packet $\File_{n,b}\in\Qbf$,  
the correlation-aware grouping procedure computes a {\em $\delta$-ensemble} $\Omega_{U_{n,b}}$, where $\Omega_{U_{n,b}}$ is the union of $\File_{n,b}$ and the subset of all cached and requested packets in $\Cbf\cup\Qbf$ that are $\delta$-correlated with $\File_{n,b}$, as per the following definition.

\begin{definition}({\bf $\bm\delta$-Correlated Vectors})\label{def:delta-cor vec}
	Consider two random vectors ${\bf X}=[X_1, \ldots, X_L]$ and ${\bf Y}=[Y_1, \ldots, Y_L]$ with i.i.d. entries such that $(X_i, Y_i) \sim p_{X,Y}(x,y)$, and $H(X)=H(Y)$.  For a given threshold $\delta\leq1$, we say that $\bf X$  is $\delta$-correlated with $\bf Y$ if 
	$H(X,Y)\leq (1+\delta) H(X)$. 
\end{definition}

\begin{remark}\label{remark1}
	For the original and updated libraries generated as in Sec.~\ref{sec:Problem Formulation}, and any $\delta<1$, it follows from Definition~\ref{def:delta-cor vec} that packets ${W}_{n,b}$ and ${W}_{n',b'}$ are never $\delta$-correlated when $b\neq b'$, and for $b=b'$ they are $\delta$-correlated if $H({ W}_{n},{ W}_{n'}) \leq (1+\delta)  H(W)$. Similarly $\delta$-correlation can be defined across packets ${\File}_{n,b}$ and ${\File}_{n',b}$, and across packets $W_{n,b}$ and ${\File}_{n',b}$.
\end{remark}

{We note that $\delta$  is a system parameter that can be optimized
	as a function of the system parameters $K,N,M$, and the library distribution, in
	order to minimize the rate of the overall CA-CACM scheme as quantified in \eqref{average-rate}.\footnote{In practice, the designer shall determine the level at which to compute and exploit correlations between files based on performance-complexity trade-offs.} 
	In fact, based on Definition \ref{def:delta-cor vec},  higher $\delta$, allows for more packets to be $\delta$-correlated, but with higher refinements.}

{\bf Correlation-Aware Group Coloring.} 
After computing the $\delta$-ensembles for each requested packet, the multicast encoder computes the multicast codeword by following a group coloring procedure on an {\em augmented} conflict graph. 
For completeness, we first describe the conventional index coding conflict graph. A conventional index coding conflict graph, which is the complement of the side information graph \cite{birk1998informed,IndexCoding}, is constructed based on the packet-level cache configuration $\Cbf$ and demand realization $\Qbf$ 
in such a way that there is one vertex for each requested packet in $\Qbf$, and there is an edge between any two vertices if one of the packets is not  cached by the receiver requesting the other packet. 
Such construction ensures decodability of an index code constructed by first computing a valid coloring of the conflict graph (see Definition \ref{def:color}), and then XORing the packets with the same color. As explained next, the augmented conflict graph is constructed by augmenting the conventional index coding conflict graph with the $\delta$-ensembles computed by the grouping procedure. 
\subsubsection*{I) Augmented Conflict Graph}
The augmented conflict graph $\mathcal H_{\Cbf,\Qbf}=(\Vc, \Ec)$ is made up of  a vertex set $\Vc$  and an edge set $\Ec$:
\begin{itemize}
	\item Vertex set $\Vc$: The vertex set $ \Vc=\Vc_r\cup\widetilde{\Vc}$ is composed of root nodes
	$\Vc_r$ and virtual nodes $\widetilde{\Vc}$. 
	\begin{itemize}
		
		\item Root Nodes: There is a root node $ v_r\in\Vc_r$ for each requested packet  $\File_{n,b} \in  \Qbf$, uniquely identified by the triplet  $\Big( \rho(v_r),\mu(v_r), r(v_r)\Big)$,  with $\rho(v_r)$ denoting the packet
		identity, $\File_{n,b}$, 
		and  $\mu(v_r)$ denoting the receiver requesting it. The third component $r( v_r) = v_r$ is introduced for ease of exposition when presenting the scheme with the virtual nodes defined next. 
		\item Virtual Nodes: For each root node $v_r \in\Vc_r$,  all the packets in the $\delta$-ensemble
		$\Omega_{\rho(v_r)}$ other than $\rho(v_r)$ are represented as virtual nodes in $\widetilde{\mathcal V}$. We
		identify virtual node $\tilde v\in\widetilde{\Vc}$, having $v_r$ as a root note, with the triplet $\Big( 
		\rho(\tilde v), \mu(v_r),r(\tilde v)\Big)$, where $\rho(\tilde v)$ indicates the 
		identity of the $\delta$-correlated packet associated with 
		vertex $\tilde v$, $\mu(v_r)$ indicates the receiver requesting $\rho(v_r)$, and  $r(\tilde v) = v_r$ is the
		root of the $\delta$-ensemble to which $\tilde v$ belongs.
	\end{itemize}
	Based on the classification of vertices into root nodes and virtual nodes, the vertex set $\mathcal V$ is partitioned into groups, such that there is one group for each root node. 
	Given a root node $v_r$, its {\em group}, denoted by $\Gc_{v_r} \subseteq \Vc$, is composed of root node $v_r$ and all virtual nodes corresponding to the packets in its $\delta$-ensemble  $\Omega_{\rho(v_r)}$. In other words, group $\Gc_{v_r}$ represents packet $\rho(v_r)$ and all packets $\delta$-correlated with it.

	\item Edge set $\Ec$: For any pair of vertices $v_1, v_2 \in \Vc $, there is an edge between $v_1$ and $v_2$ in $\Ec$  if  both $v_1$ and  $v_2$ are in the same group, or if the two following conditions are jointly satisfied:  1) $\rho(v_1) \neq
	\rho(v_2)$; 2) packet $\rho(v_1) \notin \Cbf_{\mu(v_2)}$ or packet $\rho(v_2) \notin \Cbf_{\mu(v_1)}$.  The basic idea is that a pair of vertices are connected with an edge, if their corresponding packets conflict (interfere) with each other, and hence, can not be sent together in a single transmission.
	
\end{itemize}
\begin{remark} 
	\label{remark2} 
	From the construction of the augmented conflict graph, it immediately follows that the conventional index coding conflict graph
	is the subgraph of $\mathcal H_{\Cbf,\Qbf}$ resulting from considering only $\Vc_r$.
\end{remark}
\begin{remark}\label{remark4}  If correlation is not considered or non-existent, each group is only composed of the root node, and hence $\mathcal V = \Vc_r$. Then, the augmented conflict graph is equivalent to the conventional index coding conflict graph \cite{ji2017order}.
\end{remark}

\begin{example}\label{ex:coloring}  
	Consider a network with $K=3$ receivers, cache capacity $M=3$, and a dynamic library composed of $N=6$ uniformly popular files $\{W_1^F,W_2^F,\dots,W_6^F\}$, with $H(W)=1$, and update probability vector ${\bm \pi}=(1,0,\dots,0)$. We assume that the files within each pair $\{W_1^F ,W_2^F\}$ and $\{W_5^F, W_6^F\}$ are correlated, while all other file are independent. Each file is divided into $B=4$ packets, and the packets of correlated files are $\delta$-correlated as described in Remark \ref{remark1}. The packet-level cache configuration is given in  Fig.~\ref{fig:examples} (a). During the delivery phase, an updated version of the first file becomes available which is correlated with $W_1^F$, and independent of all other files. For demand $\dbf = (1,3,5)$ the correlation-aware grouping component computes the following $\delta$-ensembles: 
	\begin{align}
	&\Omega_{\File_{1,1}} = \{\File_{1,1},W_{1,1}\}, \; \Omega_{\File_{1,2}} = \{\File_{1,2},W_{1,2}\}, \;\Omega_{\File_{1,3}} = \{\File_{1,3},W_{1,3}\}, \; \Omega_{\File_{1,4}} = \{\File_{1,4},W_{1,4}\}, \notag\\
	&\Omega_{\File_{3,1}} = \{\File_{3,1}\},\; \Omega_{\File_{3,2}} = \{\File_{3,2}\},\; \notag\\
	& \Omega_{\File_{5,2}} = \{\File_{5,2},W_{6,2}\},\; \Omega_{\File_{5,4}} = \{\File_{5,4},W_{6,4}\}  ,\notag
	\end{align}
	where for the updated library we have $\File_{3,1}=W_{3,1}$, $\File_{3,2}=W_{3,2}$, $\File_{5,2}=W_{5,2}$, and $\File_{5,4}=W_{5,4}$.  
	Given the packet-level cache and demand configurations, the root node set, $\Vc_r = \{v_1, \dots, v_8\}$, is composed of
	
	\noindent \resizebox{\linewidth}{!}{ 
		\begin{minipage}{\linewidth}
			\begin{align}
			& v_1: \Big( \File_{1,1},\,  1,\, v_1\Big),
			\, v_2:  \Big( \File_{1,2},\,  1,\, v_2\Big) , 
			\, v_3: \Big( \File_{1,3},\,  1,\, v_3\Big) ,   
			\, v_4: \Big( \File_{1,4},\,  1,\, v_4\Big) , \notag\\ 
			& v_5: \Big( \File_{3,1},\,  2,\, v_5\Big)  ,
			\,v_6: \Big( \File_{3,2},\,  2,\, v_6\Big) , \notag\\
			& v_7: \Big( \File_{5,2},\,  3,\, v_7\Big) , 
			\, v_8: \Big( \File_{5,4},\,  3,\, v_8\Big)  , \notag			
			\end{align}
		\end{minipage}
	}
	The corresponding augmented conflict graph $\mathcal H_{\Cbf,\Qbf}$ with vertices $\Vc =\Vc_r\cup
	\{\tilde v_{1},\tilde v_{2},\tilde v_{3},\tilde v_{4},\tilde v_{7},\tilde v_{8}\}$ is shown in Fig. \ref{fig:examples} (b), where for the virtual nodes we have
	\begin{align}
	&\tilde v_1: \Big( W_{1,1},\,  1,\, v_1\Big),
	\, \tilde v_2:  \Big( W_{1,2},\,  1,\, v_2\Big) , 
	\, \tilde v_3: \Big( W_{1,3},\,  1,\, v_3\Big) ,  
	\, \tilde v_4: \Big( W_{1,4},\,  1,\, v_4\Big) , \notag\\
	& \tilde v_7: \Big( W_{6,2},\,  3,\, v_7\Big)  ,
	\,\tilde v_8: \Big( W_{6,4},\,  3,\, v_8\Big) . \notag
	\end{align}

\end{example}

\begin{figure*}
	\begin{minipage}[b]{.5\linewidth}\centering
		\includegraphics[width=0.75\textwidth]{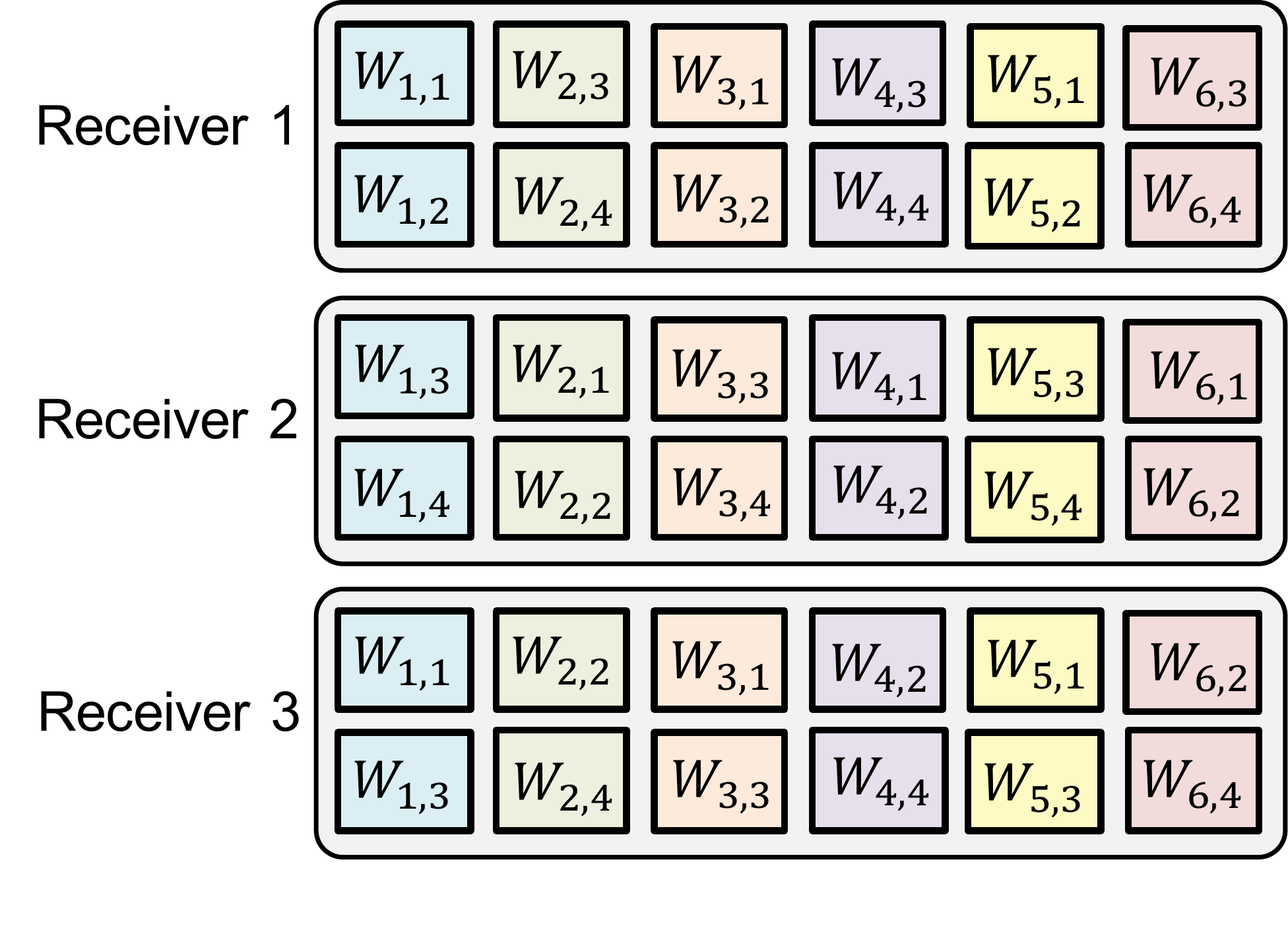}
		\subcaption{Cache Configuration}
	\end{minipage}
	\begin{minipage}[b]{.5 \linewidth}\centering
		\includegraphics[width=0.85\textwidth]{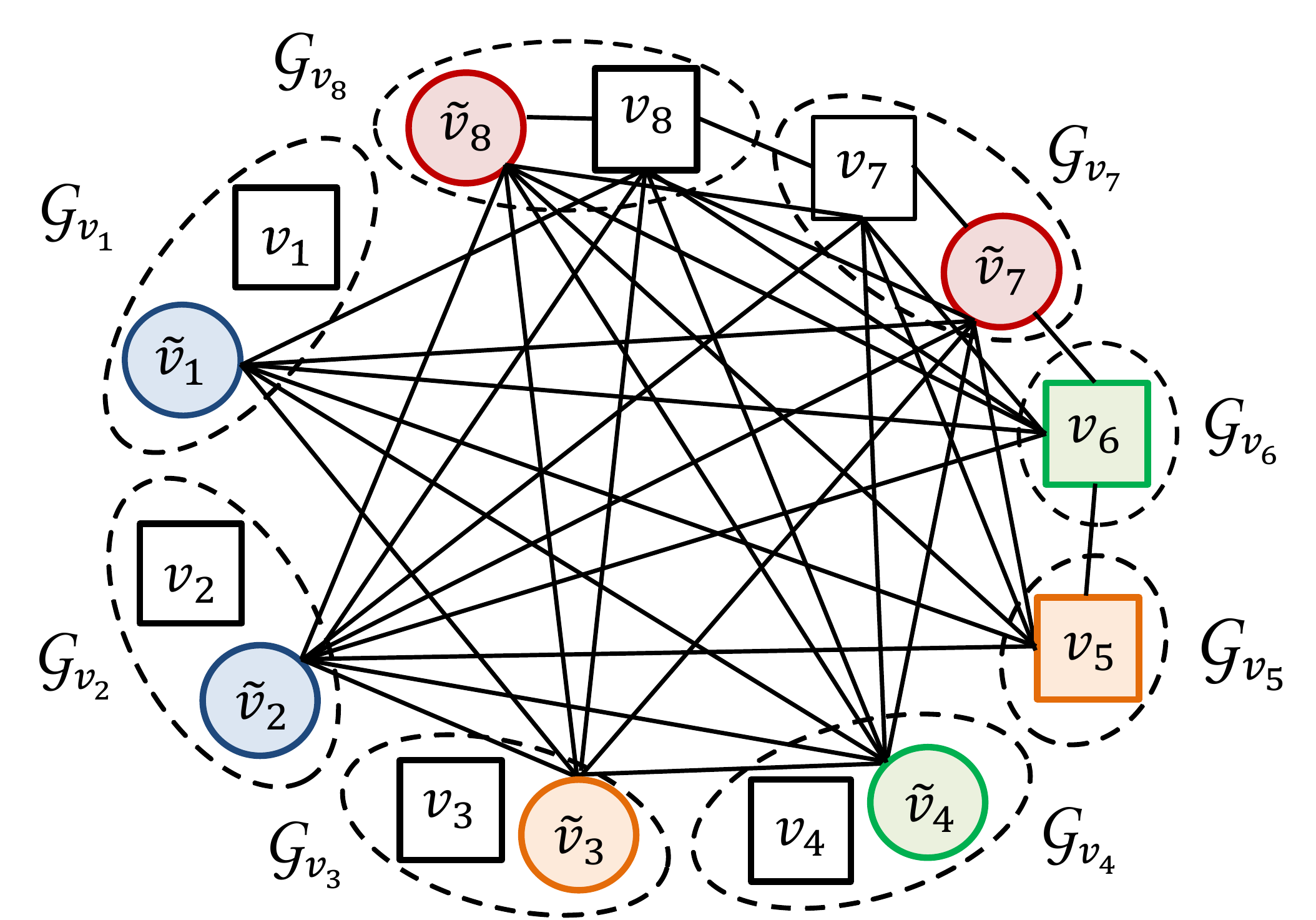}
		\subcaption{Correlation-Aware Group Coloring}
	\end{minipage}
	\captionsetup{width=.95\linewidth}
	\caption{Example \ref{ex:coloring}: (a) Packet-level cache configuration. (b) Coloring of the augmented conflict graph. Groups are delineated with dashed ovals, root nodes are represented by squares, and circles have been used for virtual nodes. Groups $\Gc_{v_1},\dots, \Gc_{v_4}$ correspond to receiver 1's demand,  groups $\Gc_{v_5}$ and $\Gc_{v_4}$, correspond to receiver 2, and groups $\Gc_{v_7}$ and $\Gc_{v_8}$, correspond to receiver 3. Vertices  $v_{1},  v_{2}$, $ v_{3}$, and $ v_{4}$  are connected to all other vertices, and all corresponding edges have been removed in the depicted graph. Each group is colored with the same color as the colored vertex shown in that group.  
	}	\label{fig:examples}
\end{figure*}

\subsubsection*{II) Linear Index Coding via Group Coloring}\label{subsec:index code}
From the augmented conflict graph $\mathcal H_{\Cbf,\Qbf}$, the multicast encoder computes a linear index code, whose construction is described next and relies on the  following two definitions. 

\begin{definition} ({\bf Graph Coloring})\label{def:color}
	A valid coloring of a graph is an assignment of colors to the vertices of the graph such that no two adjacent vertices are assigned the same color.
\end{definition}

\begin{definition} \textbf{({\bf Group Coloring})}
	Given a valid coloring of the augmented conflict graph $\mathcal H_{\Cbf,\Qbf}$, a valid   group coloring  of  $\mathcal H_{\Cbf,\Qbf}$  consists of
	assigning to each group $\Gc_{v_r},\; \forall v_r\in  \Vc_r$, one of the colors assigned to the vertices inside
	that group.
\end{definition}	

The multicast encoder proceeds as follows:
\begin{itemize}
	
	\item[1)] It finds a valid coloring of  graph $\mathcal H_{\Cbf,\Qbf}$.
	
	\item[2)] For the given graph coloring, it finds a valid group coloring  of  $\mathcal H_{\Cbf,\Qbf}$.
	
	\item[3)] For the given group coloring, from each group, it extracts the (root or virtual) vertex that has the same color as that assigned to the group. Since in the augmented conflict graph all nodes within a group are adjacent, a distinct color is assigned to each of them, and hence, only one vertex is extracted from each group. The packets corresponding to the extracted vertices are used as references for delivering the requested packets (corresponding to the root nodes). If a virtual node is selected from a group, a refinement will be required in step 5.
	
	\item[4)] For each color in the group coloring of $\mathcal H_{\Cbf,\Qbf}$, the multicast encoder XORs the packets corresponding to the vertices extracted from the groups with that color. Note that while the packet associated to a root node $v_r$, which is requested by receiver $\mu(v_r)$, is not locally available in its cache, the receiver may have cached the  packet associated to a virtual node $\tilde v \in \Gc_{v_r}$. Then, if the extracted node is a virtual node that is already cached at receiver $\mu(v_r)$, there is no need for it to be delivered, i.e., to be XORed with the other same-color packets. The concatenation of the XORed packets, denoted by $Y^{CM}_{\dbf}$, is referred to as the {\em coded} segment.
	
	\item[5)]  	For each receiver, the multicast encoder, if needed, computes a refinement codeword. The length of this refinement	is equal to the conditional entropy of the requested file given the decodable information from the coded segment and the cached information, and it is at most $\delta H(W) F$ bits. 
	The concatenation of the (uncoded) refinement codewords needed by all receivers, denoted by $Y_{\dbf}^{UM}$, is referred to as the {\em refinement} segment and is of length 
	$$L(Y_{\dbf}^{UM}) = \sum_{k=1}^K H(\File_{d_k}^F| Y_{\dbf}^{CM},\Cbf_k).$$
	
	\item[6)] The coded multicast codeword results from	concatenating the coded segment $Y^{CM}_{\dbf}$ and the refinement segment  $Y^{UM}_{\dbf}$. The overall transmission rate corresponding to the graph coloring and group coloring being considered for demand $\dbf$ is $ L(Y_{\dbf}^{CM})/F +  L(Y_{\dbf}^{UM}) /F$.
	
	\item[7)] Among all valid graph colorings, and among all corresponding valid group colorings, the multicast encoder selects the coded multicast codeword  resulting in the minimum overall transmission rate, and multicasts it over the shared link.
	
\end{itemize}

At each receiver, the multicast decoder uses the cached information and the received coded segment $Y_{\dbf}^{CM}$ to reconstruct a (possibly) distorted version of the receiver's requested packets, due to the potential reception of packets that are $\delta$-correlated with the requested ones. Then, if needed, the multicast decoder uses the refinement segment $Y_{\dbf}^{UM}$ to losslessly reconstruct its demand.

From  Remark \ref{remark2}, it follows that for each given graph coloring, there is always a  valid group coloring, whose group colors are the same as their root node. Transmission based on this group coloring is equivalent to that resulting from the coloring of the conventional index coding conflict graph. That is, among all coded multicast codewords constructed by following steps 1-6, there is also the correlation-unaware codeword corresponding to the conventional conflict graph.  Hence, the correlation-aware multicast encoder achieves lower rate than that achieved via  coloring of the conventional conflict graph. Clearly, constructing multicast codewords based on not only the requested packets, but also their correlated packets, i.e., adding virtual nodes by augmenting the conflict graph, increases available coding opportunities. In other words, computing the index code based on the augmented conflict graph allows the sender to multicast coded messages that are simultaneously useful to a larger number of receivers  compared to the conventional conflict graph.

\subsection{Greedy Group Coloring (GGC)}\label{subsec: algorithms}
In Sec.~\ref{subsec:index code}, we described how to design a coded multicast codeword that is useful to a large number of receivers based on a group coloring of the augmented conflict graph. The multicast encoder determines the optimal multicast codeword across all valid graph colorings and their corresponding group colorings. Given that graph coloring, and by extension group coloring, is NP-Hard,  in this section we propose {\em Greedy Group Coloring (GGC)}, a polynomial-time  approximation to the group coloring problem. GGC extends the existing Greedy Constraint Coloring (GCC) scheme \cite{ji2017order} to account for file correlation in cache-aided networks. 
The proposed GGC algorithm computes the coded multicast codeword by choosing between two coloring schemes, referred to as GGC$_1$ and GGC$_2$, and described in Algorithms \ref{algorithm1} and \ref{algorithm2}, respectively. Given the two valid group colorings resulting from GGC$_1$ and GGC$_2$, the coded multicast codewords of each Algorithm are computed by concatenating the corresponding coded and refinement segments. GGC chooses the scheme resulting in the lower overall rate (i.e., shorter multicast codeword). 
It is worthwhile noticing that GGC$_2$ is nothing other than the correlation-aware version of naive (uncoded) multicasting,\footnote{Naive multicasting refers to the transmission of a common, not-network-coded stream of data packets, simultaneously received and decoded by multiple receivers.} which we have included here, for the sake of completeness, in a form symmetric to that of GGC$_1$. 


\begin{algorithm}[ht]
	\caption{ GGC$_1$ - Greedy Group Coloring 1} \label{algorithm: GClC 1} \small{
		\begin{algorithmic}[1]
			\WHILE {$\Vc_r \neq \emptyset$}
			\STATE Pick any root node $v_r \in \Vc_r$
			
			\STATE Sort $\Gc_{v_r}$ in decreasing order of receiver label size\footnotemark, such that for $v_t,v_{t+1} \in  \Gc_{v_r}$, $|\{\mu(v_t),\eta(v_t)\}| \geq |\{\mu(v_{t+1}),\eta(v_{t+1})\}|$ where $v_t$ denotes the $t^{th}$ vertex in the ordered sequence.
			\STATE $t = 1$
			\WHILE {$t \leq | \Gc_{v_r}|$}
			\STATE Take $ v_t \in \Gc_{v_r}$; Let $I_{v_t} = \{ v_t\}$
			\FORALL{$ v \in \Vc\backslash \Big\{ \Gc_{v_r} \cup I_{v_t}  \Big\}$}
			\IF{ $\Big\{$ There is no edge between $v$ and $I_{v_t}$ $\Big\}$ $\bigwedge$
				$\Big\{ \{\mu(v),\eta(v)\} = \{\mu(v_t),\eta(v_t)\} \Big\} $ }
			\STATE $I_{v_t} = I_{v_t} \cup \{v\}$
			\ENDIF
			\ENDFOR
			\STATE ${v_t}^{*} = \argmax\limits_{\{ v_\tau:\, \tau=1, \ldots,t \}} |I_{v_\tau}|$
			\IF{ $|I_{v_t^{*}}| \geq |\{\mu(v_{t}),\eta(v_{t})\}|$  or $t=| \Gc_{v_r}|$ }
			\STATE $\Ic = I_{v_t^{*}} $
			\STATE $t=| \Gc_{v_r}|+1$
			\ELSE
			\STATE $t = t+1$
			\ENDIF
			\ENDWHILE
			\STATE Color all vertices in $\Ic$ with an unused color.
			\STATE {$\Jc = \{ r(v) : v \in \Ic\}$}
			\STATE { $  \Jc_r = \Big \{v_r \!\in \!\Vc_r :  \exists  v \! \in \!\Ic, \, \! \Big\{ \mu(v_r)=\mu(v) \Big\} \! \bigwedge  \!  \Big\{\rho(v) \in \Omega_{\rho(v_r)}  \Big\}     \Big \}$}
			
			\STATE $\Vc_r  \leftarrow \Vc_r\backslash  \Jc \cup \Jc_r$, $\Vc \leftarrow \Vc\backslash \bigcup\limits_{v_r \in \Jc\cup \Jc_r} \Gc_{v_r}$
			\ENDWHILE
			
		\end{algorithmic}
	}
	\label{algorithm1}
\end{algorithm}
\footnotetext{Packets with equal label size are ordered such that $H({\rho(v_r)}|{\rho(v_t)})\leq H({\rho(v_r)}|{\rho(v_{t+1})})$.}  

\begin{algorithm}[ht]
	\caption{  GGC$_2$ - Greedy Group Coloring 2} \label{algorithm: GClC 2} \small{
		\begin{algorithmic}[1]
			\WHILE {$ \Vc_r \neq \emptyset$} \STATE Pick any root node $v_r \in \Vc_r$
			\FORALL{$v \in  \Gc_{v_r}$ }
			\STATE $I_{v} = \{v\} \cup \{v' \in \Vc\backslash \Gc_{v_r} : \rho(v')=  \rho(v)\}$
			\ENDFOR
			\STATE $v^{*} = \argmax\limits_{v} |I_{v}|$
			\STATE $\Ic =  I_{v^*}$
			\STATE Color all vertices in $\Ic$ with an unused color.
			\STATE {$\Jc = \{ r(v) :  v \in \Ic\}$ }
			\STATE $\Vc_r \leftarrow \Vc_r  \backslash \Jc$, $\Vc \leftarrow \Vc\backslash  \bigcup\limits_{v_r\in \Jc} \Gc_{v_r}$
			\ENDWHILE
			
		\end{algorithmic}
	}
	\label{algorithm2}
\end{algorithm}


As mentioned previously, any vertex (root node or virtual node) $v\in\mathcal V$ is uniquely identified by the triplet $\Big(\rho(v),\mu(v), r(v)\Big)$. In GGC, for any $v\in \Vc$, we further define $ \eta(v)\triangleq\{k: \rho(v) \in \Cbf_k\}$ as the set of receivers who have cached packet $\rho(v)$.   We refer to the unordered set of receivers $\{\mu(v), \eta(v)\}$ as the {\em receiver label} of vertex $v$, which corresponds to the set of receivers either requesting or caching packet $\rho(v)$.

\begin{definition}
	{An independent set is a set of vertices in a graph, no two of which are adjacent.}
\end{definition}
Algorithm GGC$_1$ starts from a root node $v_r\in \Vc_r$ among those not yet selected, and searches for the node $v_t \in \Gc_{v_r}$ which forms the largest independent set $\Ic$  with all the vertices in $ \Vc$ having its same receiver label. 
Next, vertices in set $\Ic$ are assigned the same color (see lines 20-23).  Algorithm GGC$_2$ is based on a correlation-aware extension of GCC$_2$ in \cite{ji2017order},  and corresponds to a generalized uncoded (naive) multicast. For each root node  $v_r\in \Vc_r$, whose group has not yet been colored, only the vertex  $ v_t\in \Gc_{v_r}$ whom is found among the nodes of more groups, i.e., correlated with a larger number of requested packets, is colored, and its color is assigned to $\Gc_{v_r}$ and to all groups containing $v_t$. For both GGC$_1$ and GGC$_2$,  when the graph coloring algorithm terminates, only a subset of the graph vertices, $\Vc$, are colored, such that only one vertex from each group in the graph is colored. This is equivalent to identifying a valid group coloring where each group is assigned the color of its colored vertex. For both GGC$_1$ and GGC$_2$, the packets corresponding to the same-color vertices are XORed together and then concatenated to make the coded segment of the multicast codeword.

\begin{example}\label{ex:coloring 2}	 
	For the augmented conflict graph given in Example~\ref{ex:coloring}, a valid group coloring based on Algorithm GGC$_1$ is shown in Fig.~\ref{fig:examples} (b). Based on the greedy group coloring, the groups are colored with $4$ distinct colors, such that one color is assigned to virtual nodes $\tilde v_1$ and $\tilde v_2$, one color is assigned to virtual node $\tilde v_3$ and root node $v_5$ , one color is assigned to virtual node $\tilde v_4$ and root node $v_6$, and finally one color is assigned to virtual nodes $\tilde v_7$ and $\tilde v_8$.  Since the packets corresponding to $\tilde v_1$ and $\tilde v_2$, and the packets of $\tilde v_7$ and $\tilde v_8$  are cached at receivers 1 and  3, respectively, they are not XORed together, and not included in the multicast codeword. Based on this correlation-aware group coloring, the sender transmits a coded segment $Y^{CM}_{\dbf}  = \{W_{1,3}\oplus W_{3,1},W_{1,4} \oplus W_{3,2}\}$ 	followed by a refinement segment $Y_{\dbf}^{UM}$, with rate
	\begin{align}
	\frac{1}{F}L(Y^{CM}_{\dbf})& = H({ \File}_{1,1}|{
		W}_{1,1})+H(\File_{1,2}|{W}_{1,2})+H(\File_{1,3}|{W}_{1,3})+H({\File}_{1,4}|{ W}_{1,4})\notag\\
	&\quad+H({
		W}_{5,2}|{ W}_{6,2})+H({
		W}_{5,4}|{ W}_{6,4}) ,\notag\\
	& =  H(\File_1|W_1)  + \frac{1}{2} H(W_5|W_6) \leq \frac{3}{2}\delta. \notag
	\end{align}
	The overall load is at most $\frac{1}{2}+\frac{3}{2}\delta$. The correlation-unaware CACM scheme in \cite{ji2017order} constructs the conventional index coding conflict graph, which only contains the root nodes, and multicasts $\{\File_{1,1},\File_{1,2},\File_{1,3},\File_{1,4},$  $W_{3,1}\oplus W_{5,4},W_{3,2},W_{5,2} \}$, with rate $\frac{7}{4}$.
\end{example}

\section{ Performance of CA-CACM  with Random  Cache Placement} \label{Sec: upper bound}
{This section provides an upper bound on the rate achieved by the proposed CA-CACM scheme, denoted by $\Rach(M)$, under the assumption of random fractional caching. By considering the correlation-aware random fractional cache encoder described in Sec. \ref{subsec:CA RAP}, we are able to evaluate $\Rach(M)$ for a system with arbitrary number of receivers and files, as $F,B\rightarrow\infty$, which also upper bounds $R^*(M)$.}

\subsection{Static Library} \label{subsec:rate static}
In this subsection, we provide an upper bound on the rate achieved with the proposed CA-CACM scheme, and hence, on the optimal rate-memory function, for a static setting as described in Sec.~\ref{sec:Problem Formulation}. In order to quantify the performance of the proposed scheme, we consider  a symmetric library distribution, such that,  for a given $ \delta$ and for any packet, the size of its $\delta$-ensemble, i.e.,  the set of packets in the library that are $\delta$-correlated with it, has cardinality equal to $G_\delta$.  Note that due to the symmetry across library files, the optimal caching distribution is uniform across files, and ${\varrho_n} = \frac{1}{N}$ for $n\in\{1,\dots,N\}$.

\begin{theorem}\label{thm:uniform rate}
	Consider a  static broadcast caching network with $K$ receivers, library size $N$, cache capacity $M$, and uniform demand distribution. Then, the rate-memory function, $R^*(M)$, is upper bounded as
	\begin{align}
	R^*(M) \leq   \limsup_{F\rightarrow\infty} {  \Rach(M)}\leq  
	\inf_{\delta}\,\min \Big\{\Psi_1^S(K, N,{ M,\delta}) , \, \Psi_2^S( K, N,{  \delta}) \Big\} , 
	\label{R_GGC}
	\end{align}
	where
	\begin{align}
	&\Psi_1^S(K, N,M,\delta)=  \sum\limits_{\ell =1 }^{K} \binom{K}{\ell}    \Big(1-\frac{M}{N} \Big)  \Big(  \frac{M}{N} \Big)^{G_\delta-1}  
	\bigg[
	P_{\ell}       + 
	\sum\limits_{g =1}^{G_\delta-1} \binom{G_\delta-1}{g}  \Big(\frac{N}{M}-1 \Big)^{g}  \phi(\ell,g)  \bigg], \\
	& \Psi_2^S( K, N,\delta) =   (1-\delta)\, \Phi\Big( K,      \frac{N}{G_\delta}   \Big)     + \delta    \,\Phi( K, N)          , \label{eq:Nbar}
	\end{align}
	with
	\begin{align}
	& \phi(\ell,g) =  (\widehat P_{\ell}) ^{g+1}   \bigg( \alpha(\ell, g+1) \;+\;\sum\limits_{t=1}^{g} \;\alpha(\ell,t)  \; \Big( \frac{P_{\ell}}{\widehat P_{\ell}} \Big)^t  
	\bigg[   \binom{g+1}{t} +   \delta \, \xi(\ell)   \, \binom{g}{t}  \bigg]  \bigg)  \label{eq:phi} ,\\
	& \alpha(\ell,t)  =  \hspace{-3mm}\sum\limits_{d=1}^{\min\{t,\,x(\ell)\}}
	\frac{1}{d}  \binom{x(\ell)-1}{d-1} \bigg( \sum\limits_{  t_1+  \dots+ t_d = t  }    \frac{ t! }{t_1! \,\dots\, t_d!  }  \bigg) ,    \label{eq:alpha}\\
	&\Phi(\kappa, \nu) = \nu   \, \Big(   1- \Big(1-\frac{1}{\nu}\Big)^{\kappa}     \Big), \label{naivemult}\\
	&  P_{\ell} =   \Big(1-\frac{M}{N}\Big)^{(K-\ell) }\Big(\frac{M}{N}\Big)^{(\ell-1) }, \\
	&  \widehat P_{\ell} =   \sum_{i=1}^{\ell-1} \binom{K-1}{i-1}   P_{i}  ,\;\; \ell =2,\dots ,K,\; \quad   \widehat P_{1} = 0,\\
	& \xi (\ell) = \sum\limits_{i=1}^{\ell}\,i\, \binom{\ell}{i}  ( \widehat P_{\ell}    )^i   (P_{\ell})^{\ell-i}  ,\\
	&  x(\ell) = \binom{K-1}{\ell-1}.
	\end{align}
\end{theorem}

\begin{proof}
	The proof is given in Appendix \ref{app:uniform rate}.
\end{proof}

Note that the rate  in Theorem~\ref{thm:uniform rate} is obtained by first deriving an upper bound, i.e., 
$$\min\{\Psi_1^S(K, N,M,\delta) , \; \Psi_2^S(K, N,\delta) \},$$
for a given  $\delta$, and then minimizing it with respect to $\delta$. 

As stated in Remark~\ref{remark4}, if each group is only composed of its root node, then the augmented conflict graph is equivalent to the conventional index coding conflict graph. Therefore, in the special case of selecting $\delta$ such that $G_\delta=1$,   $\min\{\Psi_1^S(K, N,M,\delta) , \; \Psi_2^S(K, N,\delta) \}$ provides an
upper bound on the rate achieved with the correlation-unaware scheme proposed in \cite{ji2017order}.

\subsection{ Dynamic Library} 
This subsection provides an upper bound on the rate achieved with the proposed CA-CACM scheme for a dynamic setting as described in Sec.~\ref{sec:Problem Formulation}, when $\pi_n=\pi$ for any $n\in\{1,\dots,N\}$.

\begin{theorem}\label{thm:uniform rate dynamic}
	Consider a dynamic broadcast caching network with $K$ receivers, library size $N$, cache capacity $M$, and uniform demand distribution. For a given $\delta$ and update probability $\pi$, the rate-memory function, $R^*(M)$, is upper bounded as  
	\begin{align}
	R^*(M)    \leq& \limsup_{F\rightarrow\infty} {  \Rach(M)}\leq     \min \Big\{\Psi_1^D(K,N,M,\delta,\pi) , \Phi( K,   N  )  \Big\} , 
	\end{align}
	where 	 
	\begin{align}
	\Psi^D_1(K,N,M,\delta,\pi)  = \sum\limits_{\ell =1 }^{K} \binom{K}{\ell}  \Big(1- \frac{M}{N} \Big)P_{\ell}    +  \delta \,\Phi(  K_\pi,N_\pi ) ,
	\end{align}	
	with $\Phi(\kappa,\nu)$ and $P_{\ell}$ as defined in Theorem \ref{thm:uniform rate},  $K_\pi\triangleq\pi  K$ and $N_\pi\triangleq\pi N $. 
	
\end{theorem}

\begin{proof}
	The proof is given in Appendix~\ref{App:uniform rate dynamic}.
\end{proof}

\section{ Performance of CA-CACM with Deterministic Cache Placement} \label{Sec: 2user2file}
This section quantifies the rate achieved by the proposed CA-CACM scheme, denoted by  $\Rach(M)$, under the assumption of deterministic  fractional caching. {Specifically, we focus on  a simple  static broadcast caching network with two receivers and two correlated files, for which we are able to deterministically design the close-to-optimal cache configuration.  For this network, a  Gray-Wyner based correlation-aware CACM scheme was proposed in \cite{ISITjournal} and proven to be optimal over a region of the memory.}

We assume that the library is composed of two uniformly popular files $\{W_1^F,\,W_2^F\}$ generated by a 2-DMS, such that  $H(\Wsf_1|\Wsf_2)= H(\Wsf_2|\Wsf_1)=\delta H(W) $. 
Each file is divided into $B=2$ equal length packets with $\frac{1}{2}H(W)F$ bits, and based on the above assumption the packets from file 1 and 2 are $\delta$-correlated as described in Remark \ref{remark1}.  As  worst-case and best-case demands, we consider demands $\dbf_1 = (1,2)$ and $\dbf_2=(1,1)$, respectively. The proposed CA-CACM scheme operates as follows:
\begin{itemize}
	\item $M = 0$:  In this case, the receivers have not stored any of the files. Then,  based on the group coloring approach in Sec \ref{subsec:CA Delivery},
	
	\begin{itemize}
		\item For demand $\dbf_1$,  the sender delivers  one of the files, for example $W_1^F$, and sends a refinement with rate $H(\Wsf_2|\Wsf_1)$ so that the second receiver can recover file $W_2^F$.   
		\item For demand $\dbf_2$,  the sender multicasts file $W_1^F$.
	\end{itemize}	
	Hence, the proposed CA-CACM scheme simply results in the correlation-aware version of  naive multicasting, which leads to  an average load of
	
	\noindent\resizebox{0.95\linewidth}{!}{
		\begin{minipage}{\linewidth}
			\begin{align} 
			\Rach(M) &=  \frac{1}{4}\Big(H(\Wsf_1) +{ H(\Wsf_2)} +2H(\Wsf_1,\Wsf_2)\Big)   = \Big(1 +\frac{\delta}{2}\Big)H(W).\notag
			\end{align}
		\end{minipage}
	}

	\item $M = H(W)$:  The caches are filled as
	\begin{align}
	&Z_{1} = \{W_{1,1},\; W_{2,2}\},\;\; Z_{2} = \{ W_{1,2},\; W_{2,1} \}. \label{cache1}
	\end{align}
	\begin{itemize}
		\item	 {For the worst-case demand $\dbf_1$, the CA-CACM scheme constructs the augmented conflict graph based on the demand and cache contents. The graph consists of two groups $\mathcal G_{v_1}$ and $\mathcal G_{v_2}$, with roots $v_1$ and $v_2$, respectively, such that
			\begin{align}
			&\mathcal G_{v_1} =\{v_1,\tilde v_1\},\; v_1: \Big(\File_{1,2},1,v_1\Big),\; \tilde v_1: \Big(W_{2,2},1,v_1\Big) ,\notag\\
			&\mathcal G_{v_2} =\{v_2,\tilde v_2\},\; v_2: \Big(\File_{2,2},2,v_2\Big),\; \tilde v_2: \Big(W_{1,2},2,v_2\Big) ,\notag
			\end{align}
			where, due to the static setting considered, $\File_{1,2} = W_{1,2}$ and $\File_{2,2} = W_{2,2}$. 
			Table~\ref{tb:graph2}  displays four valid  group colorings of the augmented conflict graph. 
			As explained in Sec.~\ref{subsec:CA Delivery}, among the group colorings shown in Table~\ref{tb:graph2}, the multicast encoder selects the one resulting in the lowest overall delivery rate, i.e., Colorings 1 or 2  depending on the value of $\delta$. Note that in Coloring 2, since both packets corresponding to the extracted vertices (colored with blue) are already cached at their corresponding receivers, they are not transmitted. Similarly, for Colorings 3 and 4, the locally available packets $W_{2,2}$ and $W_{1,2}$, respectively, are not included in the XORed codeword.}
		
		\item  For demand $\dbf_2$, the sender builds the augmented conflict graph as described in Sec \ref{subsec:CA Delivery}, and identifies two group colorings: $i)$ A coloring that 
		corresponds to transmitting two refinements with rates $\frac{1}{2}H(\Wsf_{1}|\Wsf_{2})$ and 
		$\frac{1}{2}H(\Wsf_{2}|\Wsf_{1})$, so that receiver 1 can recover packet $W_{1,2}$ from  $W_{2,2}$, and receiver 2 can recover $W_{1,1}$ from $W_{2,1}$, respectively. And, $ii)$ a coloring that, as in the correlation-unaware scheme, corresponds to transmitting the XOR of the requested packets missing from each receiver's cache. 
		The multicast encoder selects the group coloring resulting in the lowest overall delivery rate, given as
		\begin{align}
		\min\Big\{    \frac{1}{2}  H(\Wsf)  ,\, \frac{1}{2}H(\Wsf_1|\Wsf_2)+ \frac{1}{2}H(\Wsf_2|\Wsf_1)    \Big\}  
		= 
		\min\Big\{\frac{1}{2},\, \delta\Big\} H(W).  \notag
		\end{align}
	\end{itemize}
	Therefore, the  average load for capacity $M=H(W)$ is 
	\begin{align}
	\Rach(M)  &=   
	\min\Big\{\frac{1}{2},\, \delta\Big\} H(W).  \notag
	\end{align}	
	
	\item  $M = 2H(W)$:  The library is fully stored at both receivers resulting in zero average load.
\end{itemize}

\noindent With the proposed caching policy, the memory-rate pairs 
$$\Big(M,\,R\Big)\in\bigg\{\Big(0,\, (1+ \delta/2)H(W) \Big),\;\Big( H(W) ,\,  \min \{1/2 ,\, \delta \} H(W)   \Big),\Big(  2H(W), 0 \Big)\bigg\}$$
are achievable, and, as in \cite{maddah14fundamental}, through memory-sharing we achieve the lower convex envelope of these points. Hence, 
\begin{align} 
{ \Rach}(M) =            
\begin{cases}
(1+\frac{\delta}{2}) (H(W)-M) + \min\{\frac{1}{2}, {\delta}\}  M,                
& M \in  [0,\,H(W)] \\
\min\{\frac{1}{2},{\delta} \}   (2H(W)-M)  ,        
&     M \in  (H(W),\, 2H(W)]
\end{cases}\label{eq: rate 2 user 2file}
\end{align}

\begin{table}
	\centering
	\includegraphics[width= 0.7\linewidth]{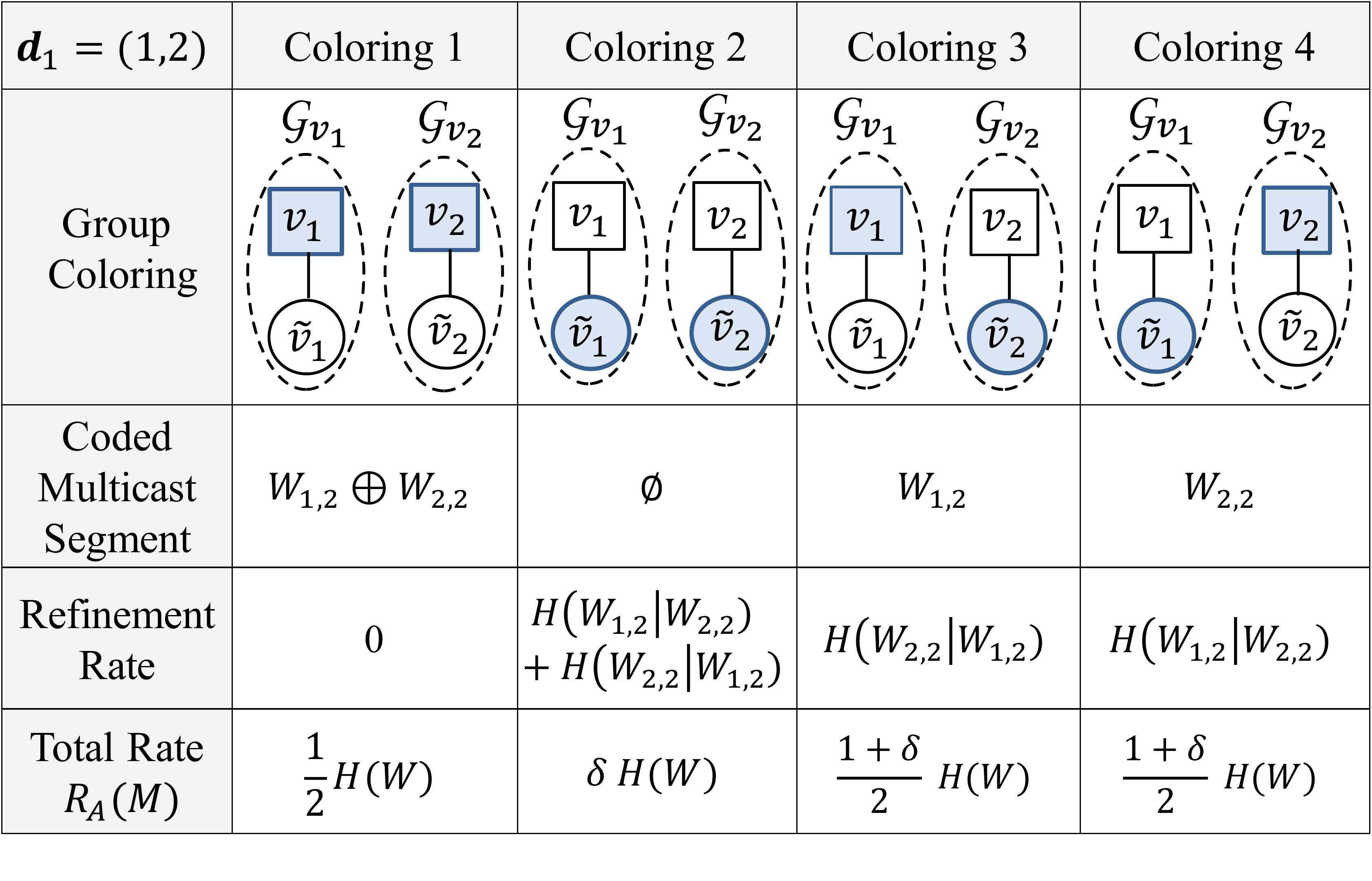}
	\captionsetup{width=.98\linewidth}
	\caption{Correlation-aware group coloring for the two-receiver two-file system with deterministic caching 
		when $M=H(W)$.}
	\label{tb:graph2}
\end{table}

The following theorem quantifies the rate gap of the proposed scheme from the optimal rate-memory trade-off, $R^{*}(M)$, as a function of the cache capacity M.
\begin{theorem}\label{thm:optimality of scheme} 
	For any cache capacity $M \in [0,H(W)]$, the proposed CA-CACM scheme achieves an average rate $\Rach(M)$, such that
	$$ \Rach(M) - R^*(M) \,\leq   \frac{1}{2} \min \Big\{H(\Wsf_1|\Wsf_2)  ,\, I(\Wsf_1;\Wsf_2) \Big\}      ,$$
	and for any  $M \in (H(W),\; 2H(W)]$, the scheme performs within half of the mutual information, i.e.,
	$$ \Rach(M) - R^*(M) \,\leq     \frac{1}{2} I(\Wsf_1;\Wsf_2)      .$$	
\end{theorem}
\noindent\begin{proof}
	The proof is given in Appendix~\ref{app:optimality of scheme}, which follows from comparing $\Rach(M)$ given in eq.~\eqref{eq: rate 2 user 2file} with the lower bound given in \cite{ISITjournal}. 
\end{proof}

We remark that the proposed correlation-aware CACM scheme described above places content in the receiver caches while taking into account the correlation among the files. Suppose that for $M=H(W)$, the caches were filled as 
\begin{align}
&Z_{1} =  \{W_{1,1},\; W_{2,1}\},\;\; Z_{2} = \{ W_{1,2},\; W_{2,2}\}.  \label{cache2}
\end{align}
In a conventional correlation-unaware scheme as in \cite{maddah14fundamental,maddah14decentralized,ji2017order,ji14groupcast}, and \cite{ji15multiple}, the cache configurations in \eqref{cache1} and \eqref{cache2} yield the same average load of $\frac{1}{2}H(W)$. However, the proposed CA-CACM scheme achieves a lower rate when adopting cache configuration \eqref{cache1}, as it provides better references for compression during the delivery phase.

We numerically compare the performance of the proposed CA-CACM scheme with the lower bound given in \cite{ISITjournal}, and with the Gray-Wyner based CACM scheme proposed in \cite{hassanzadeh2017rate,ISITjournal}, which uses  Gray-Wyner source coding \cite{gray1974source} to jointly compress the files before the caching phase. Fig.~\ref{fig:two files} displays the rate-memory trade-off for a 2-DMS with $H(\Wsf) =1$ and $H(\Wsf_1|\Wsf_2)= H(\Wsf_2|\Wsf_1)=\delta=0.25$, which, as stated in Theorem~\ref{thm:optimality of scheme}, performs very close to the lower bound.

\begin{figure} [H]
	\centering
	\includegraphics[width=0.5\textwidth]{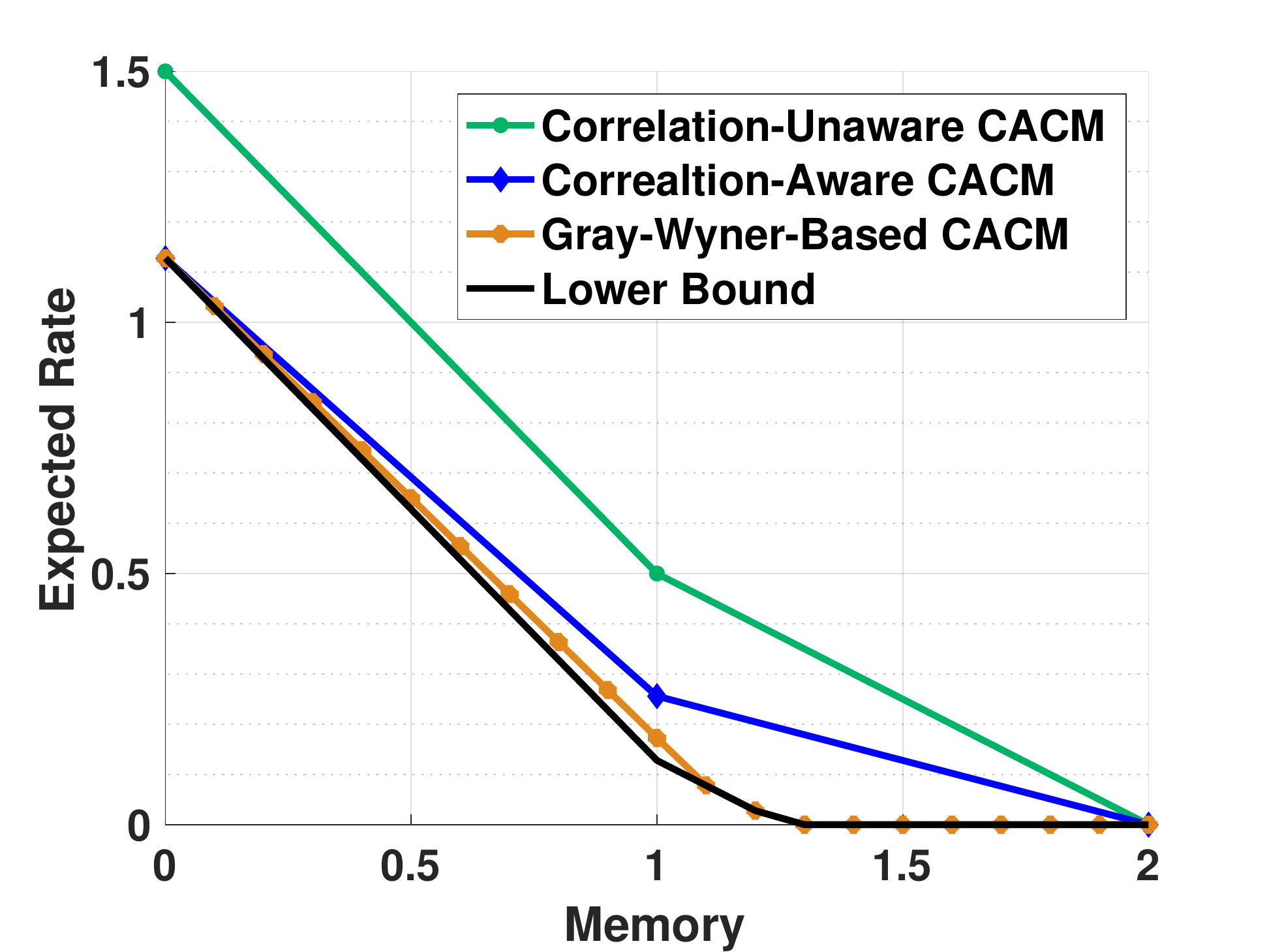}
	\captionsetup{width=.98\linewidth}
	\caption{Rate-memory trade-off for the two-receiver system for files generated according to a 2-DMS with $H(W) = 1$ and $\delta = 0.25$.}
	\label{fig:two files}
\end{figure}

\section{Numerical Results and Discussions}~\label{sec:Simulations}
We numerically compare the rates achieved by the proposed correlation-aware scheme,  CA-CACM, with
respect to the state-of-the-art CACM scheme proposed in \cite{ji2017order}, which does not consider and exploit content correlation. The scheme in \cite{ji2017order} is a  combination of random fractional caching and coded multicasting, and, as stated in Sec.~\ref{subsec:rate static}, it is equivalent to the proposed CA-CACM when $G_\delta=1$, and consequently, when the files are independent.  

\subsection{Static Library:}
We consider a broadcast caching network with $K=10$ receivers requesting files from a static library, as described in Sec.~\ref{sec:Problem Formulation}. Figs. \ref{fig:AllFigs}(a) and (b) display the rate-memory trade-off as the memory size varies for a library with symmetric statistics such that $H(W)=1$, and for $\delta =0.1$, each file (and consequently each of its packets) is $\delta$-correlated with $G_\delta$ other files. The figures plot the expected achievable rate, in average number of transmissions normalized to the file size, versus memory size (cache capacity) $M$, normalized to the file size for a library with $N=20 $ and $90$ files. As expected, the correlation-aware scheme outperforms the scheme that is oblivious to the file correlations.

\begin{figure*}
	\begin{subfigure}{0.5\linewidth}\centering
		\includegraphics[width=0.75\linewidth]{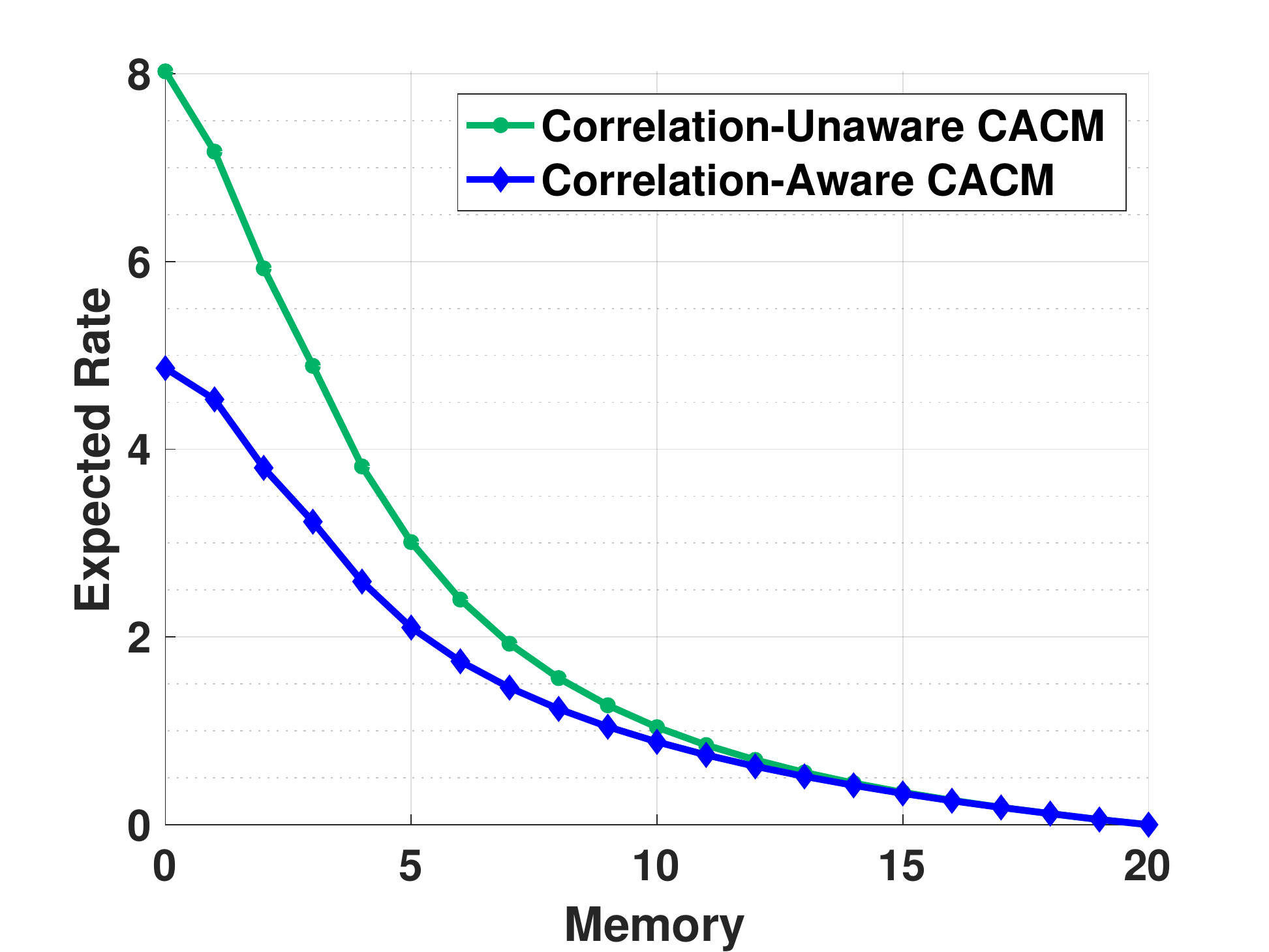}
		\subcaption{Static Setting: $K = 10$, $N = 20$, $\delta = 0.1$, and $G_\delta = 4$.}
	\end{subfigure}\hspace*{\fill}
	\begin{subfigure}{0.5\linewidth}\centering
		\includegraphics[width=0.75\linewidth]{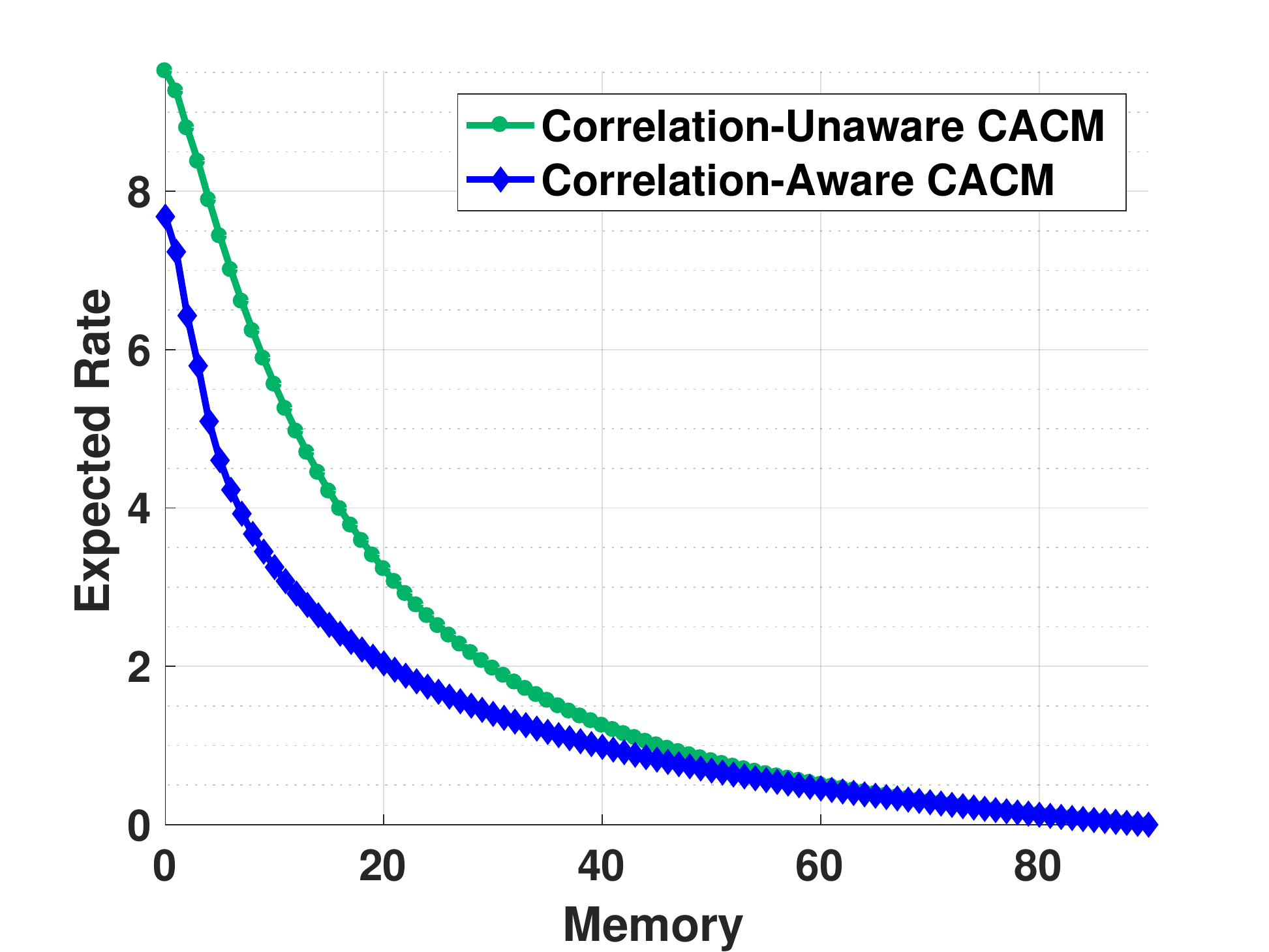}
		\subcaption{Static Setting: $K = 10$, $N=90$, $\delta = 0.1$, and $G_\delta = 6$.}
	\end{subfigure}
	\begin{subfigure}{0.5\linewidth}\centering
		\includegraphics[width=0.8\linewidth]{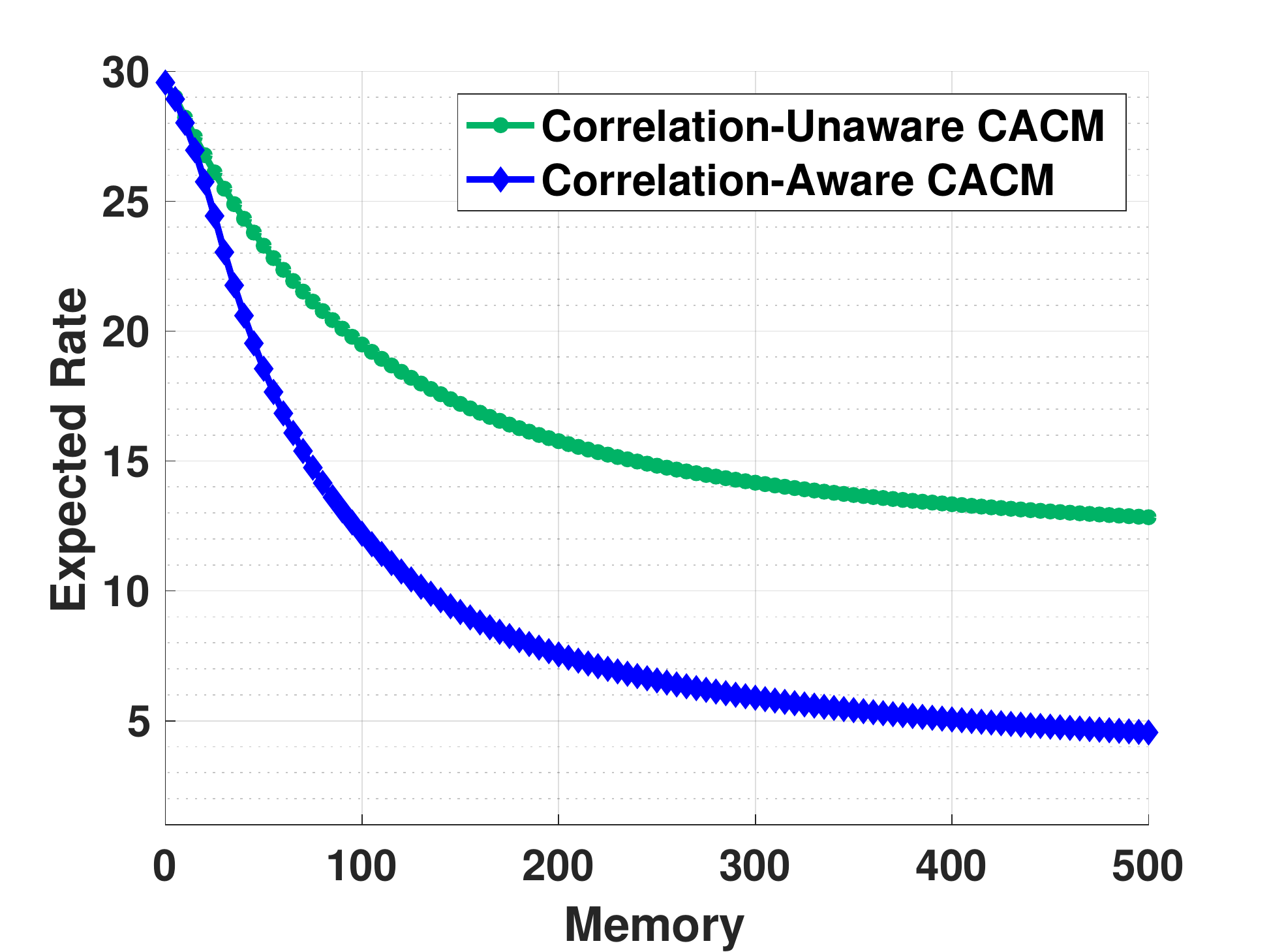}
		\subcaption{Dynamic Setting: $K = 30$, $N=1000$, $\delta = 0.3$, $\pi = 0.4$.}
	\end{subfigure}\hspace*{\fill}
	\begin{subfigure}{0.5\linewidth}\centering
		\includegraphics[width=0.8\linewidth]{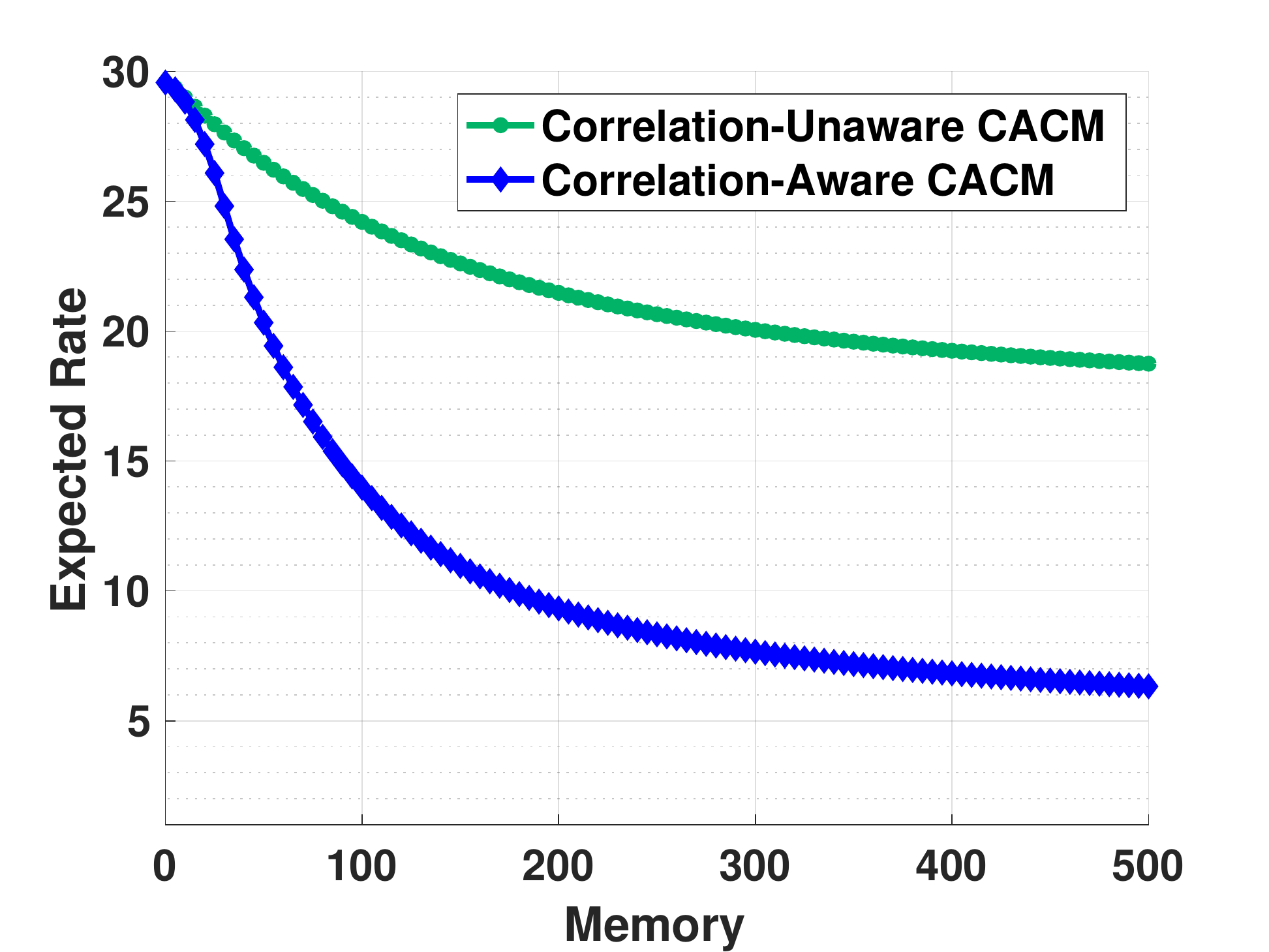}
		\subcaption{DynamicSetting: $K = 30$, $N=1000$, $\delta = 0.3$, $\pi = 0.6$.}
	\end{subfigure}
	\captionsetup{width=.98\linewidth}
	\caption{Rate-memory trade-off for a library with: (a), (b) static correlated content, and (c), (d) dynamic correlated content.}
	\label{fig:AllFigs}
\end{figure*}

We observe that the correlation-aware scheme is able to achieve rate reductions that go well beyond the state of the art correlation-unaware counterpart. The improvement in performance relies on the fact that, during the delivery phase, the sender compresses the set of requested files into a multicast codeword composed of content correlated to the requested files in addition to the requested files themselves, which results in increased coding opportunities. Specifically, the proposed  CA-CACM  achieves a $1.56\times$ reduction in the expected rate compared to the correlation-unaware scheme for a cache capacity equal to $10\%$ of the library size when $N=20$, and when $N = 90$ there is a $1.7\times$ reduction.

Next, we assume that the receivers are equipped with a cache capacity for storing a percentage of the library files, and compare the performance of CA-CACM with the correlation-unaware scheme. In Fig.~\ref{fig:RateUser} the schemes are compared for a library with $N=1000$ files, while the number of receivers varies from $K = 10$ to $50$ when each receiver's cache capacity is $10\%$ of the entire library size. We remark that the results displayed here correspond to a given system parameter $\delta$, and we do not optimize over the choice of $\delta$. 

\begin{figure}[H]
	\centering
	\includegraphics[width=0.5\linewidth]{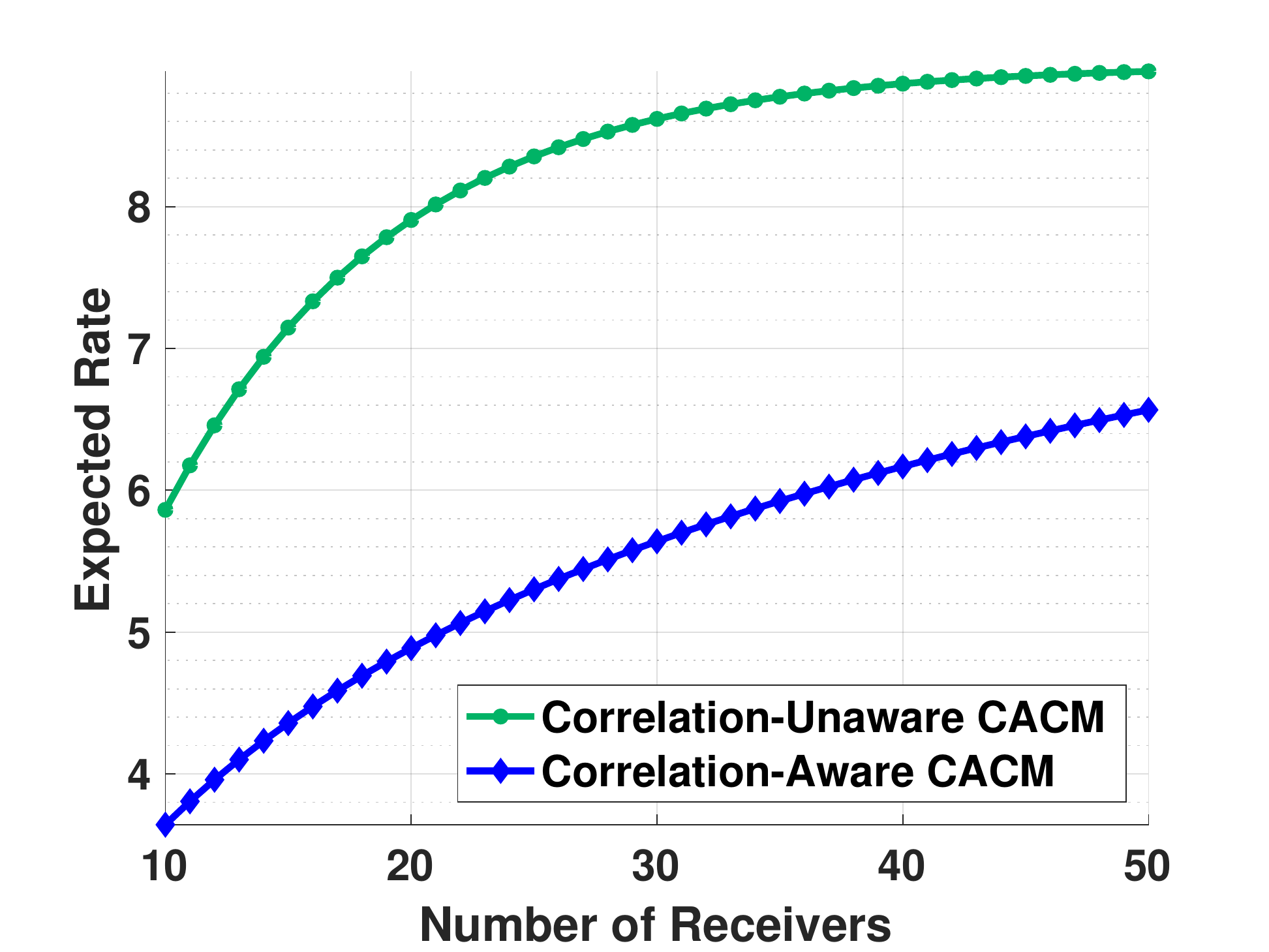}
	\captionsetup{width=.98\linewidth}
	\caption{Rate-memory trade-off in a network with $N = 1000$, $M$ $=0.1\,N$, $\delta = 0.1$, $G_\delta = 5$ for number of receivers from $10$ to $50$. }
	\label{fig:RateUser}
\end{figure}

\subsection{Dynamic Library:}
We consider a network with dynamically correlated content, as described in Sec.~\ref{sec:Problem Formulation}.  Figs.~\ref{fig:AllFigs}(c) and (d) illustrate the potential of the proposed CA-CACM scheme for exploiting the correlation among the different versions of the content. The plots are given for a system with $K=30$ receivers, $N=1000$ files, and a level of correlation of $\delta = 0.3$ among the different versions of a file.

From the plots it is observed that the CA-CACM scheme  is significantly superior in utilizing the cached content to alleviate the load on the shared link compared to the correlation-unaware scheme. When new versions of files become available at the sender with a probability of $\pi =0.4$, the CA-CACM scheme performs $2.8\times$ better compared to the correlation-unaware scheme for a cache size equal to half the library, and its performance gap becomes $3\times$ when updated versions become available with a larger probability of $\pi =0.6$. The CA-CACM scheme is steadily effective in exploring global caching gains,  across all cache capacities, among the files that have not been updated as well as the files for which a new version is available. However, since the new content has not been prefetched in the receiver caches, the correlation-unaware scheme falls short from constructing coded messages for delivering the requested content for which a new version is available, and can only use coding for the files that have not changed.

\section{Conclusion}~\label{sec:Conclusion}
In this paper, we have formulated the problem of efficient delivery of dynamic correlated sources over a broadcast caching network using information-theoretic tools. 
We have proposed a correlation-aware scheme in which receivers store content pieces based on their popularity as well as on their correlation with the rest of the file library in the caching phase, and receive compressed versions of the requested files according to the information distributed across the network and their joint statistics during the delivery phase. 
The proposed scheme is shown to significantly outperform state of the art approaches that treat library files as mutually independent content. 
In addition, when compared with recent solutions that employ joint file compression before the caching phase, our scheme, by using individually compressed files in the caching phase and providing on-demand compression during the delivery phase, can adapt to changes in the content library, yielding superior robustness to system dynamics, particularly relevant in next generation dynamic content services.

\begin{appendices}
	\section{Proof of Theorem \ref{thm:uniform rate}} \label{app:uniform rate}
	We upper bound the expected rate achieved by the proposed CACM scheme for a given $\delta$, as $F,B\rightarrow\infty$, where the group coloring of the augmented graph is done based on the Greedy Group Coloring (GGC) described in Sec.~\ref{subsec: algorithms}. Recall that the expected rate achieved by the scheme is the sum of the coded segment rate resulting from Algorithms \ref{algorithm: GClC 1}   and \ref{algorithm: GClC 2}, and the uncoded refinement segment rate that is required for lossless reconstruction of the demands. Once the total rate is quantified, minimizing over $\delta$ results in the right-hand side of \eqref{R_GGC}. This provides an upper bound on the rate achieved with the Group Coloring based CACM scheme proposed in Sec.~\ref{subsec:index code}, which has NP hard complexity, and in turn, upper bounds the optimal rate-memory function $R^{*}(M)$. We prove Theorem \ref{thm:uniform rate} following the same procedure as in \cite[Appendix A]{ji2017order}.
	\subsection{Rate Achieved by Algorithm GGC$_1$}
	\subsubsection{\bf Coded Segment Rate}\label{subsec:uni multicast}
	For a given packet-level cache configuration $\Cbf$ and demand realization $\Qbf$,
	the  rate achieved by GGC$_{1}$ is a function of the total  number of colors  assigned by Algorithm \ref{algorithm: GClC 1} to the augmented conflict graph $\mathcal H_{\Cbf,\Qbf}$, which we denote by $\Jc(\Cbf, \Qbf)$. By definition $\Jc(\Cbf, \Qbf)$ is  the total number of independent sets {\em selected} by the algorithm. Recall that an independent set is a set of vertices in a graph, no two of which are adjacent. 
	
	Let $\Kc_{\ell}\subseteq \Kc \equiv \{1,\dots,K\}$ denote a subset of $\ell$ receivers. By construction, Algorithm~\ref{algorithm: GClC 1}  associates to each subset $\Kc_{\ell}$ a total number of $\Jc_{\Cbf,\Qbf}( \Kc_{\ell})$ independent sets of size $\ell$, such that for each independent set $\Ic$
	\begin{itemize}
		\item[1)] $\forall v\in \mathcal I,\;   \{\mu(v), \eta(v)\}  \equiv \Kc_{\ell}$,
		\item[2)] $\forall v\in \mathcal I,\;  \nexists v'\in \Gc_{r(v)} $  belonging to an independent set of size larger than $\ell$.
	\end{itemize}
	We refer to $\Kc_{\ell}$ as the {\em receiver label} of independent set $\Ic$.  As a result, when considering all possible receiver labels  $\Kc_{\ell}$ of size $\ell\in \{1,\dots, K\}$, each group in the graph will only be assigned to independent sets of a given size (i.e., independent sets of different size are not assigned to the same group.). However, multiple  independent sets of the same size can be assigned to one group, in the event of which, we, uniformly at random, only select one of those sets. We remark that this random selection is done for analytical evaluation of the number of colors assigned in the augmented conflict graph, and is equivalent to line (3) in Algorithm ~\ref{algorithm: GClC 1}. 
	
	Consequently, in order to quantify $\Jc(\Cbf, \Qbf)$, we generate all possible receiver labels, and count the total number of independent sets associated to each receiver label. A necessary condition for an independent set to be associated to receiver label $\Kc_{\ell}$ is that, for any receiver $k \in \Kc_{\ell}$, there exist a group $\Gc_{v_r}$ with root node $v_r \in \Vc_r$, such that
	\begin{itemize}
		\item[1)] $\mu(v_r)= k$, i.e., receiver $k$ is requesting packet $\rho(v_r)$.
		\item[2)] There exists a node $v\in\Gc_{v_r}$, such that $\eta(v)=\Kc_{\ell}\setminus \{k\}$, i.e., $\rho( v)$ is cached by all receivers in $\Kc_{\ell}\setminus \{k\}$, and
		not by any other receiver.
		\item[3)] $\Kc_{\ell}$ is the largest receiver label across all $v\in\Gc_{v_r}$ that satisfies conditions 1 and 2.
	\end{itemize}
	
	Then, for a given $\Cbf$ and $\Qbf$, the number of independent sets {\em selected} by the algorithm becomes
	
	\begin{equation}
	\Jc(\Cbf,\Qbf) = \sum\limits_{\ell =1 }^{K} \sum\limits_{\Kc_{\ell} \subseteq \Kc }   \mathcal J_{\Cbf,\Qbf}(\Kc_{\ell}) 
	\end{equation}
	with
	\begin{equation}
	\mathcal J_{\Cbf,\Qbf}(\Kc_{\ell}) =  \max\limits_{k\in \Kc_{\ell}}  \; \sum\limits_{\substack{ v_r \in \Vc_r: \\ \rho(v_r) \ni \Qbf_k  }  }        \mathbbm{1}    \Big\{  \mathcal \Kc_{\ell}   \text{ is associated to }    \Gc_{v_r}    \Big\},   
	\end{equation}
	where, for brevity, by ``$\Kc_{\ell}$ is associated to $\Gc_{v_r} $'' we mean that $\Gc_{v_r} $ is assigned to an independent set  associated with receiver label  $\Kc_{\ell}$. We define the random variable,
	$$\Ysf(\Kc_{\ell}, \Gc_{v_r}) \triangleq \mathbbm{1}    \Big\{  \mathcal \Kc_{\ell}   \text{ is associated to }    \Gc_{v_r}    \Big\}.$$
	
	Enumerating over all receiver labels, the expected number of independent sets assigned by Algorithm~\ref{algorithm: GClC 1} is
	\begin{align}
	\EE_{\Cbf}\bigg[ \EE_{\Qbf}\Big[   \mathcal J(\Cbf,\Qbf )  |    \Cbf   \Big]  \bigg]   
	&=   \EE_{\Cbf}\Bigg[\EE_{\Qbf}\bigg[ \sum\limits_{\ell =1 }^{K} \sum\limits_{\Kc_{\ell} \subseteq \Kc }   \mathcal J_{\Cbf,\Qbf}(\Kc_{\ell})    \Big|   \Cbf \bigg]\Bigg]\notag\\
	& =    \EE_{\Cbf}\Bigg[  \EE_{\Qbf} \bigg[ \sum\limits_{\ell =1 }^{K} \sum\limits_{\Kc_{\ell} \subseteq \Kc }   \max\limits_{k\in \Kc_{\ell}}   \sum\limits_{\substack{ v_r \in \Vc_r: \\ \rho(v_r) \ni \Qbf_k  }  }       \hspace{-2mm}   \Ysf(\Kc_{\ell}, \Gc_{v_r})    \Big|  \Cbf \bigg]   \Bigg] \notag\\ 
	& =   \sum\limits_{\ell =1 }^{K} \sum\limits_{\Kc_{\ell} \subseteq \Kc }   \EE_{\Cbf}\Bigg[   \EE_{\Qbf} \bigg[   \max\limits_{k\in \Kc_{\ell}}    \sum\limits_{\substack{ v_r \in \Vc_r: \\ \rho(v_r) \ni \Qbf_k  }  } \hspace{-2mm}      \Ysf(\Kc_{\ell}, \Gc_{v_r})    \Big|   \Cbf \bigg]   \Bigg],    \label{eq: colors}
	\end{align}
	where $\Cbf$ is the random cache configuration resulting from the
	random caching scheme with uniform caching distribution, and $\Qbf$ is the packet-level demand realization resulting from the random i.i.d. requests with uniform demand distribution. 
	
	Based on Lemma~\ref{Lemma:Lemma 1} given in Appendix \ref{App:Lemma 1}, for a given $\Gc_{v_r}$ and $\Kc_{\ell}$, the random variable $\Ysf(\Kc_{\ell}, \Gc_{v_r}) $ follows a Bernoulli distribution with parameter
	\begin{align}
	\hspace{-2mm}\lambda(\ell,G_\delta) =    \sum\limits_{g =0 }^{G_\delta-1} \binom{G_\delta-1}{g}  \Big(1-\frac{M}{N} \Big)^{g} \Big(\frac{M}{N} \Big)^{ G_\delta-1 - g } \psi(\ell,g) ,  
	\end{align}
	where $G_\delta$ is the size of group $\Gc_{v_r}$ for any root node $v_r \in  \Vc_r$, and $\psi(\ell,g) $ is given in \eqref{eq:psi}. Similar to \cite[Appendix A]{ji2017order} it can be shown that for $B\rightarrow \infty$, we have
	\begin{align}
	\max\limits_{k\in \Kc_{\ell}}  \; \sum\limits_{\substack{ v_r \in \Vc_r: \\ \rho(v_r) \ni \Qbf_k  }  }   \; \frac{\Ysf(\Kc_{\ell}, \Gc_{v_r}) }{B(1-M/N)}\;  \stackrel{p}{\rightarrow}\; \lambda(\ell,G_\delta)   .
	\end{align}
	
	Hence, from \eqref{eq: colors}, as $F,B\rightarrow\infty$, the expected coded multicast rate is given by
	\begin{align}
	\frac{1}{B}\; \EE_{\Cbf}\bigg[ \EE_{\Qbf}\Big[   \mathcal J(\Cbf,\Qbf )   \;| \;   \Cbf   \Big]  \bigg] 
	=  \sum\limits_{\ell =1 }^{K} \binom{K}{\ell}   \Big(1-\frac{M}{N}\Big)  \;\lambda(\ell,G_\delta)     . 
	\end{align}

	\subsubsection{\bf Refinement Segment Rate}\label{subsec:uni unicast}
	Additional transmissions are required to enable lossless reconstruction of the requested packets for which a correlated packet was delivered rather than the packet itself. We quantify the number of such packets using an approach similar to the multicast rate.   From Lemma~\ref{Lemma:Lemma 1}, when considering receiver label $\Kc_\ell$, a refinement is required when a {\em virtual} node belongs to an independent set with receiver label $\Kc_{\ell}$, and it is selected for transmission rather than the root node. A necessary condition for one of the virtual nodes to be a candidate for transmission is that at least one of the corresponding packets not be cached at the receiver requesting the root node. Therefore, the range of $g$ in Lemma~\ref{Lemma:Lemma 1} becomes $\{1,\dots,G_\delta-1\}$. When $g\in \{1,\dots,G_\delta-1\}$ of the $\delta$-correlated packets are not cached at the receiver, a refinement is required only when the root node does not belong to an independent set with receiver label $\Kc_{\ell}$. As explained in Lemma~\ref{Lemma:Lemma 1}, this condition is equivalent to the root node belonging to an independent set with smaller size, while at least one of the remaining $g$ virtual nodes constructs $\Kc_{\ell}$, which occurs with probability
	\begin{align}
	\widehat P_{\ell} \; \Big( \sum\limits_{t=1}^{g} \binom{g}{t}     \; (P_{\ell})^{t} \; (\widehat P_{\ell})^{g-t}  \Big),
	\end{align}
	with $P_{\ell}$ and $\widehat P_{\ell}$ defined in \eqref{eq:Pell} and \eqref{eq:P hat App}.
	
	Recall that when an independent set of size $\ell$ is selected by the Algorithm, all the nodes in the independent set are XORed (coded) together and sent as one transmission over the shared link. By construction the selected independent set can be composed of virtual nodes rather than the root node, in the event of which additional uncoded transmissions are delivered for each virtual node that is in the independent set, each with rate $\delta$. A node belongs to an independent set with receiver label size $\ell$ with probability $P_{\ell}$, and it belongs to a set with smaller size with probability $\widehat P_{\ell}$. Therefore, the expected number of uncoded transmission in an independent set of size $\ell$ is
	\begin{align}
	\xi(\ell) = \sum\limits_{i=1}^{\ell}   \,i\, \binom{\ell}{i}  (\widehat P_{\ell})^i   (P_{\ell})^{\ell-i} . 
	\end{align}

	Overall, the refinement segment rate is upper bounded by
	\begin{align}
	\delta\,	\sum\limits_{\ell =1 }^{K} \binom{K}{\ell}  \Big(1-\frac{M}{N} \Big)    \;   \xi(\ell)  \; \Delta\lambda(\ell,G_\delta) ,  
	\end{align}
	with 
	\begin{align}	
	&\Delta\lambda(\ell,G_\delta)  = \sum\limits_{g =1   }^{G_\delta-1} \binom{G_\delta-1}{g}  \Big(1-\frac{M}{N} \Big)^{g} \Big(\frac{M}{N} \Big)^{ G_\delta-1 - g } \hspace{-1mm}\Delta \psi(\ell,g)  ,  \\
	&\Delta \psi(\ell,g) \triangleq \sum\limits_{t=1}^{ g}\,      \binom{g}{t}  \,   \alpha(\ell,t)     \; (P_{\ell})^{t} \,    (\widehat P_{\ell})^{g+1-t}   ,  
	\end{align}
	and $\alpha(\ell,t) $ defined as in \eqref{eq:lemma alpha}.

	\subsection{Rate Achieved by Algorithm GGC$_2$}\label{app: GGC2}
	Algorithm GGC$_2$ executes group coloring by randomly selecting a group in the augmented conflict graph, and  determining the best representative of the group, i.e., the node that results in the highest naive multicasting gain. In other words, it identifies the (root or virtual) node whose corresponding packet belongs to the largest number of groups in the graph, where, for ease of exposition, we say packet $\rho(v)$ belongs to group $\Gc_{v_r}$ if there exists a node $v'\in \Gc_{v_r}$ such that $\rho(v)=\rho(v')$. Once the representative is identified, the same color is assigned to all the groups to which the representative packet belongs, and the group representative is multicasted (by itself) over the shared link. This transmission is sent to deliver the requested packets corresponding to the root nodes of all colored groups, and additional refinements are delivered to ensure lossless recovery of the requested packets. The procedure of identifying a representative from a randomly selected group is repeated until a color has been assigned to all the groups in the augmented conflict graph. 
	
	We upper bound the achievable rate by the rate corresponding to cache capacity $M=0$. In this case, the augmented conflict graph consists of vertices representing all the packets of the requested files, and therefore, we can upper bound the rate on the file level rather than on the packet level. Based on Definition~\ref{def:delta-cor vec}, we define the {\em file-level} $\delta$-ensemble of $W_n^F$, $n\in\{1,\dots,N\}$,  as the set of files in the library that are $\delta$-correlated with it, i.e. the set 
	\begin{align}
	E_n\triangleq \Big\{ W_{n'}^F:   H(W_{n'},W_n)\leq(1+\delta)H(W), n'\in\{1,\dots,N\} \Big\}. 
	\end{align}
	Due to the symmetry across library files, the size of each file-level $\delta$-ensemble is equal to $G_\delta$. Based on GGC$_2$, if two requested files are $\delta$-correlated, i.e., each of them is in the file-level  $\delta$-ensemble of the other one, the sender selects one of the files. It mutlicasts the selected file and sends a refinement for the reconstruction of the one not selected. The normalized expected rate is upper bounded by 
	\begin{align}
	\Psi_2^S(K,N, \delta)  \leq 
	\EE\Big[ &\text{Number of selected distinct requested files} \Big]   \notag \\
	&\quad +
	\delta \; \EE\Big[ \text{Number of distinct requests not selected} \Big] .\label{eq:GCC2 1}
	\end{align}
	In order to upper bound the first term in \eqref{eq:GCC2 1}, we determine the number of distinct requested files, such that none  of them belongs to the file-level $\delta$-ensemble of another selected file. Consider file $W_n^F$ for which $|E_n|=G_\delta$, i.e., it is $\delta$-correlated with $G_\delta$ other files. File $W_n^F$ will be multicasted over the shared link if a file from its file-level $\delta$-ensemble is requested, and $W_n^F$ is selected as the representative. Since all the files are equivalent, each file is selected as the representative with probability $1/G_\delta$. The probability of requesting a file from $\delta$-ensemble $E_n$ is    
	\begin{align}
	\PP\Big(  \mathbbm 1&\Big\{ \text{A file from $E_n$ is requested}\Big\}\Big)  
	= 1- \Big(1-\frac{G_\delta}{N}\Big)^{K},  
	\end{align}
	and therefore, for the first term in \eqref{eq:GCC2 1} we have
	\begin{align}  
	\EE\Big[ \text{Number of selected distinct requested files} \Big]   
	&= \EE\bigg[ \sum\limits_{n = 1}^{N}   \, \mathbbm 1\Big\{ \text{File $W_n^F$ is selected as the representative}\Big\}\bigg] \notag \\
	&= \EE\bigg[ \sum\limits_{n = 1}^{N}   \, \frac{1}{G_\delta}\mathbbm 1\Big\{ \text{A file from $E_n$ is requested}\Big\}\bigg] \notag \\
	& = \sum\limits_{n = 1}^{N}  \frac{1}{G_\delta} \Big(1- \Big(1-\frac{G_\delta}{N}\Big)^{K}     \Big) \notag \\
	&=   \frac{N}{G_\delta} \Big(   1- \Big(1-\frac{G_\delta}{N}\Big)^{K}     \Big) .\label{eq:GGC2 2} 
	\end{align}
	For the second term in \eqref{eq:GCC2 1}, we have
	\begin{align}
	\EE\Big[ &\text{Number of distinct requests not selected} \Big] \notag\\
	&=\EE\Big[ \text{Number of distinct requests}\Big]- \EE\Big[ \text{Number of selected distinct requested files} \Big]\notag\\
	&\stackrel{(a)}{=} \EE\bigg[ \sum\limits_{n = 1}^{ N }   \, \mathbbm 1\Big\{ \text{File $W_n^F$ is requested}\Big\} \bigg]   
	-\frac{N}{G_\delta} \Big(   1- \Big(1-\frac{G_\delta}{N}\Big)^{K}     \Big)  \notag\\
	& = N\Big(   1- \Big(1-\frac{1}{N}\Big)^{K}     \Big) -    \frac{N}{G_\delta} \Big(   1- \Big(1-\frac{G_\delta}{N}\Big)^{K}     \Big)  ,\label{eq:GGC2 3} 
	\end{align}
	where (a) follows from \eqref{eq:GGC2 2}.  By combing \eqref{eq:GCC2 1}-\eqref{eq:GGC2 3}, the rate is upper bounded as
	\begin{align}
	\Psi_2^S(K,N,\delta) \leq&
	(1-\delta)      \frac{N}{G_\delta}  \Big(   1- \Big(1-\frac{G_\delta}{N}\Big)^{K}     \Big)    \;+\;  \delta\; N\Big(   1- \Big(1-\frac{1}{N}\Big)^{K}     \Big)
	\end{align}

	\section{Lemma 1}\label{App:Lemma 1}
	\begin{lemma}\label{Lemma:Lemma 1}
		Under the setting of Theorem \ref{thm:uniform rate}, for any group $\Gc_{ v_r}$ and receiver label $\Kc_{\ell}$, the random variable $\Ysf(\Kc_\ell,\Gc_{v_r}) =  \mathbbm{1}    \Big\{  \mathcal \Kc_{\ell}   \text{ is associated to }    \Gc_{v_r}    \Big\}   $ follows a Bernoulli distribution with parameter
		\begin{align}
		&	\lambda(\ell,G_\delta) =  \sum\limits_{g =0 }^{G_\delta-1} \binom{G_\delta-1}{g}  \Big(1-\frac{M}{N} \Big)^{g} \Big(\frac{M}{N} \Big)^{ G_\delta-1 - g } \psi(\ell,g)  ,   \\
		&	 \psi(\ell,g)  \triangleq \sum_{t=1}^{g+1}\; \binom{g+1}{t}     \alpha(\ell, t)   \; (P_{\ell})^t   \;  (\widehat P_{\ell} )^{g+1-t}  ,  \label{eq:psi} \\
		&\alpha(\ell,t)   \triangleq  \hspace{-3mm}\sum\limits_{d=1}^{\min\{t,\,x(\ell)\}}
		\hspace{-1mm}  \frac{1}{d}  \binom{x(\ell)-1}{d-1} \bigg( \sum\limits_{  t_1+  \dots+ t_d = t  }    \frac{ t! }{t_1! \,\dots\, t_d!  }  \bigg) ,  \hspace{-2mm}   \label{eq:lemma alpha} 
		\end{align}
		
	\end{lemma}
	
	\begin{proof}	
		Recall that by {\em $\mathcal \Kc_{\ell}$ is associated to $\Gc_{v_r}$} we mean that group $\Gc_{v_r} $ is assigned to an independent set  associated with receiver label  $\Kc_{\ell}$.
		Without loss of generality, we illustrate how $\Gc_{v_r} $ is assigned to an independent set with receiver label $\mathcal \Kc_{\ell}$, by considering $\Gc_{v_r} = \{v_1,\dots,v_{G_\delta-1}\}\cup\{v_r\}$. Based on Condition 2 in Appendix~\ref{app:uniform rate}, a necessary condition for node $v\in\Gc_{v_r}$ to be part of the selected independent set associated with $\mathcal \Kc_{\ell}$ is that packet $\rho(v)$ not be cached at receiver $\mu(v)$, i.e., $\mu(v) \notin \eta(v)$. Then, $\rho(v)$ is a candidate for transmission for satisfying receiver $\mu(v_r)$'s request. Note that by definition root node $v_r$ is not cached at the receiver requesting it, and therefore, it is a candidate for transmission. Let us assume that only $g\in\{0,\dots, G_\delta-1\}$ nodes among the $G_\delta-1$ virtual nodes are a candidate for being transmitted, for example nodes $\Vc_g=\{v_1,\dots,v_g\}$. This means that the packets corresponding to those $g$ nodes are not cached at receiver $\mu(v_r)$, while the remaining $G_\delta-1-g$ packets are cached at $\mu(v_r)$. This event occurs with probability 
		$$
		\sum_{\Vc_g  \subseteq \Gc_{v_r}  \setminus\{v_r\}} \Big(1-\frac{M}{N}\Big)^g    \Big(\frac{M}{N}\Big)^{G_\delta-1-g}.$$
		
		Group $\Gc_{v_r} $ is assigned to an independent set associated with  $\Kc_{\ell}$, if  among the $g+1$ candidate nodes in $\{v_r,\Vc_g\}$, there are $t\in\{1,\dots,g+1\}$ nodes, say $\Vc_{g,t} = \{v_1,\dots,v_t\}$, such that:
		\begin{itemize}
			\item Any $v \in \Vc_{g,t}$ is part of an independent set of size $\ell$, one of which has receiver label $\Kc_{\ell}$,
			\item And, the remaining nodes in the set $\{v_r,\Vc_g\}$  belong to independent sets of size smaller than $\ell$.
		\end{itemize}
		We denote this event by $\Ec\Big(\Vc_g,\Vc_{g,t},\Kc_{\ell}\Big)$. Then,
		\begin{align}
		\lambda(\ell,G_\delta) & \,=\, \mathbb P  \bigg(     \mathbbm{1}    \Big\{  \mathcal \Kc_{\ell}   \text{ is associated to }    \Gc_{v_r}    \Big\}     = 1\bigg)   \notag\\
		=&  \sum_{g=0}^{G_\delta-1} \; \sum_{\Vc_g  \subseteq \Gc_{v_r}  \setminus\{v_r\}} \Big(1-\frac{M}{N}\Big)^g    \Big(\frac{M}{N}\Big)^{G_\delta-1-g}    \;\sum_{t=1}^{g+1}\; \sum_{\Vc_{g,t}  \subseteq \{v_r,\Vc_{g}\} }  \PP\bigg( \Ec\Big(\Vc_g,\Vc_{g,t},\Kc_{\ell}\Big)  \bigg) . \hspace{-8mm} \label{eq:lemma lambda} 
		\end{align}
		
		Among the $t$ nodes in $\Vc_{g,t} $, if there are $d$ distinct labels of size $\ell$, one of which is receiver label $\Kc_{\ell}$, only one is selected and associated to $\Gc_{v_r}$. The event that $\Kc_{\ell}$ is selected, denoted  by $\Ec( \Kc_{\ell} \text{ is selected} )$, occurs with probability $\frac{1}{d}$, where $d=\{1,\dots, \min\{t,x(\ell)  \}  \}$, and $x(\ell) = \binom{K-1}{\ell-1}$ denotes the total number of receiver labels of size ${\ell}$ containing receiver $\mu(v_r)$. In order to compute  $ \PP\bigg( \Ec\Big(\Vc_g,\Vc_{g,t},\Kc_{\ell}\Big)  \bigg) $, we need to compute the probability of the event that each $v_i\in \Vc_{g,t} $  belongs to a specific independent set with label $\Kc_{\ell}^{(i)}$, where the set of receiver labels $\{\Kc_{\ell}^{(1)},\dots,\Kc_{\ell}^{(t)}\}$ contains $d$ distinct labels one of which is $\Kc_{\ell}$. We denote this event by $\Ec\Big(\Vc_{g,t}, \{\Kc_{\ell}^{(1)},\dots,\Kc_{\ell}^{(t)}\}, d \Big)$, which occurs with probability
		\begin{align}
		\PP\bigg(&\Ec\Big(\Vc_{g,t}, \{\Kc_{\ell}^{(1)},\dots,\Kc_{\ell}^{(t)}\}, d \Big) \bigg) =	 (  P_{\ell} )^t ,
		\end{align}
		with 
		\begin{align}
		P_{\ell} \triangleq 	  \Big(1-\frac{M}{N}\Big)^{(K-\ell) }\Big(\frac{M}{N}\Big)^{(\ell-1) } .\label{eq:Pell}
		\end{align}
		Let $\mathcal S_{\ell,t,d}$ denote the set of all possible sets $\{\Kc_{\ell}^{(1)},\dots,\Kc_{\ell}^{(t)}\}$ containing only $d$ distinct labels. By summing over $\mathcal S_{\ell,t,d}$
		\resizebox{\linewidth}{!}{
			\begin{minipage}{\linewidth}
				\begin{align} 
				&\sum\limits_{\{\Kc_{\ell}^{(1)},\dots,\Kc_{\ell}^{(t)}\}\in \mathcal S_{\ell,t,d} } 	\hspace{-2mm}
				\PP\bigg( \Ec\Big(\Vc_{g,t}, \{\Kc_{\ell}^{(1)},\dots,\Kc_{\ell}^{(t)}\}, d \Big) \bigg) 
				=
				\Big|\mathcal S_{\ell,t,d}\Big|  (P_{\ell})^t , \notag
				\end{align}
			\end{minipage}
		}
		where
		\begin{align}
		\Big|\mathcal S_{\ell,t,d}\Big|  =   \bigg( \sum_{  t_1+ \dots+ t_d = t  }    \frac{ t! }{t_1!\, t_2!\,\dots\, t _d!  }  \bigg) \binom{x(\ell)-1}{d-1} . \label{eq:Sell}
		\end{align}
		By definition,  probability $ \PP\bigg( \Ec\Big(\Vc_g,\Vc_{g,t},\Kc_{\ell}\Big)  \bigg) $ is given by
		\begin{align}
		\PP&\bigg( \Ec\Big(\Vc_g,\Vc_{g,t},\Kc_{\ell}\Big)  \bigg)  =\notag\\
		&\sum\limits_{d=1}^{\min\{t,x(\ell)\}}  \sum\limits_{\{\Kc_{\ell}^{(1)},\dots,\Kc_{\ell}^{(t)}\}\in \mathcal S_{\ell,t,d} }   \PP\bigg(\Ec\Big(\Vc_{g,t}, \{\Kc_{\ell}^{(1)},\dots,\Kc_{\ell}^{(t)}\}, d \Big)\wedge \Ec\Big( \Kc_{\ell} \text{ is selected} \Big) \bigg) 
		\Gamma\Big(  \{v_r,\Vc_g\}\setminus \Vc_{g,t}, \ell\Big), \label{eq: blah}
		\end{align}
		where $\Gamma \Big( \{v_r,\Vc_g\} \setminus \Vc_{g,t}, \ell\Big)$ is the probability that the nodes in the set $\{v_r,\Vc_g\}\setminus \Vc_{g,t}$  belong to independent sets of size smaller than $\ell$, which is given by
		\begin{align}
		\Gamma\Big( \{v_r,\Vc_g\} \setminus \Vc_{g,t}, \ell\Big)=   ( \widehat P_{\ell})^{g+1-t} , 
		\end{align}
		with 
		\begin{align}
		\widehat P_{\ell} \triangleq \sum_{i=1}^{\ell-1} \binom{K-1}{i-1}   P_{i} , \,\ell =2,\dots,K,\quad \widehat P_{\ell}   = 0\label{eq:P hat App}
		\end{align}
		From \eqref{eq: blah}, 
		\begin{align}
		\PP\bigg( \Ec\Big(\Vc_g,\Vc_{g,t},\Kc_{\ell}\Big)  \bigg)  = 
		\sum\limits_{d=1}^{\min\{t,x(\ell)\}}  \Big|\mathcal S_{\ell,t,d}\Big|  
		\frac{1}{d}\; (P_{\ell})^t   (\widehat P_{\ell})^{g+1-t} . \label{eq: last1}
		\end{align}		
		Lemma \ref{Lemma:Lemma 1} follows by replacing \eqref{eq: last1} in \eqref{eq:lemma lambda}.
	\end{proof}

	\section{Proof of Theorem~\ref{thm:uniform rate dynamic}}\label{App:uniform rate dynamic}
	As mentioned in Sec.~\ref{sec:Problem Formulation}, during the caching phase, the receivers fill their caches from the original library, and during the delivery phase, the demands are delivered from an updated library. In the updated library, a new version of each file becomes available with probability $\pi$. We denote the random number of updated files available during the delivery phase by $\widetilde N$. Then, a random number of the receivers, denoted by $\widetilde K$, request a file that has been updated, and the remaining $K-\widetilde K$ receivers request a files for which a new version is not available. Since the original library is composed of independent files, if file $n\in\{1,\dots,N\}$ has not been updated, then the $\delta$-ensemble of its packets are of size one, i.e., $\Omega_{U_{n,b}}=\{W_{n,b}\}$, and $G_\delta=1$. The $\delta$-ensemble of the packets of a file that has been updated contains two packets, $\Omega_{U_{n,b}}=\{U_{n,b},W_{n,b} \}$, and therefore,  $G_\delta=2$.   
	
	\subsection{Coded Segment Rate Achieved with GGC$_1$:}
	We upper bound the rate achieved with Algorithm GGC$_1$, as in Appendix \ref{app:uniform rate}, by upper bounding  the rate of the coded multicast  segment with \eqref{eq: colors}. Then, for a given receiver label of size $\ell$, for example label  $\Kc_\ell$, and a given root node $v_r$ in the augmented graph:
	\begin{itemize}
		\item If packet $\rho(v_r)$ corresponds to a file that has not been updated, based on Lemma~\ref{Lemma:Lemma 1} given in Appendix \ref{App:Lemma 1}, the random variable $\Ysf({\mathcal K}_{\ell}, \Gc_{v_r}) $ follows a Bernoulli distribution with parameter $\lambda(\ell,1)=P_\ell$, with $P_\ell$ given in \eqref{eq:Pell}. Similar to \cite[Appendix A]{ji2017order} it can be shown that for $B\rightarrow\infty$
		\begin{align}
		\sum\limits_{\substack{ v_r \in \Vc_r: \\ \rho(v_r)\ni \Qbf_k }  }    \frac{ \Ysf(\Kc_{\ell}, \Gc_{v_r})}{ B(1-M/N)} \;\;\stackrel{p}{\rightarrow}\;\; P_\ell   \label{eq:not updated} 
		\end{align}
		
		\item If packet $\rho(v_r)$ corresponds to a file that has been updated, then, none of its packets have been stored in the receiver caches. Therefore, $v_r$ can only belong to an independent set of size $\ell=1$, i.e., the packet can only be sent by itself. Proceeding as in Lemma~\ref{Lemma:Lemma 1}, for $\Gc_{v_r}=\{v_r,\tilde v\}$ we have the following cases.
		\begin{itemize}
			\item[$\circ$] {\bf Case 1}: With probability $\frac{M}{N}$ packet $\rho(\tilde v)$ is cached at $\mu(v_r)$. Then $g=0$, $\Vc_g=\emptyset$, and the only candidate for transmission is node $=v_r$, which is part of an independent set of size $\ell$ with probability $\widetilde P_\ell \triangleq   {\mathbbm 1}\{\ell=1\}$. Given that $\rho(\tilde v)$ is already cached at $\mu(v_r)$, as mentioned in Sec.~\ref{subsec:index code},  packet $\rho(\tilde v)$ is not included in the multicast codeword. Therefore, this case does not contribute to the coded segment of the multicast codeword. 
			
			\item[$\circ$] {\bf Case 2}: With probability $\Big(1-\frac{M}{N}\Big)$ packet $\rho(\tilde v)$ is not cached at receiver $\mu(v_r)$, then, $g=1$, $\Vc_g =\{\tilde v\}$, and both vertices $\{v_r,\tilde v\}$ are a candidate for transmission. Consequently, group $\Gc_{v_r}$ is assigned to $\Kc_\ell$ if $t$ of the candidates, for example the set $\Vc_{g,t} $, are such that all $t$ nodes are part of an independent set of size $\ell$, one of which is $\Kc_\ell$, and the other $g-t$ nodes belong to independent sets of size smaller than $\ell$. This event is denoted by $\Ec\Big(\Vc_g,\Vc_{g,t},\Kc_{\ell}\Big)$. Therefore, 
			\begin{itemize}
				\item[-] When  $t=1$ and $\Vc_{g,t} = \{v_r\}$: Since with probability one, root node $v_r$ belongs to an independent set with size $\ell=1$, and $\tilde v$ does not belong to an independent set with smaller size, i.e., $\widehat P_\ell=0$ when $\ell=1$ (as defined in \eqref{eq:P hat App}), then $\PP\bigg(\Ec\Big(\Vc_g,\Vc_{g,t},\Kc_{\ell}\Big)\bigg) = 0$.
				
				\item[-] When  $t=1$ and $\Vc_{g,t} = \{\tilde v\}$: For any $\ell >1$, since $v_r$ always belongs to an independent set with smaller size, then $\PP\bigg( \Ec\Big(\Vc_g,\Vc_{g,t},\Kc_{\ell}\Big)  \bigg)=P_\ell \; {\mathbbm 1}\{\ell>1\} $.
				
				\item[-] When  $t=2$ and $\Vc_{g,t}  = \{v_r,\tilde v\}$: the two nodes belong to independent sets with equal size only when $\ell=1$, and therefore,  $\PP\bigg( \Ec\Big(\Vc_g,\Vc_{g,t},\Kc_{\ell}\Big)  \bigg)= P_\ell \; {\mathbbm 1}\{\ell=1\} $. 
			\end{itemize}
			\vspace{5mm}
			Combining all these events, and from \eqref{eq:lemma lambda}, group $\Gc_{v_r}$ is assigned to an independent set associated with $\Kc_\ell$ with probability
			\begin{align}
			\mathbb P  \bigg(  &   \mathbbm{1}    \Big\{  \mathcal \Kc_{\ell}   \text{ is associated to }    \Gc_{v_r}    \Big\}     = 1\bigg)  \notag\\
			&\qquad  = \sum_{g=0}^{1} \; \sum_{\Vc_g  \subseteq \{\tilde v\}} \Big(1-\frac{M}{N}\Big)^g    \Big(\frac{M}{N}\Big)^{1-g}   \;\;\sum_{t=1}^{g+1}\; \sum_{\Vc_{g,t}  \subseteq \{v_r,\tilde v\} }  \PP\bigg( \Ec\Big(\Vc_g,\Vc_{g,t},\Kc_{\ell}\Big)  \bigg) \notag \\
			& \qquad = \Big(\frac{M}{N} \Big) \times {0} + \Big(1-\frac{M}{N}\Big) \Big(P_\ell \; {\mathbbm 1}\{\ell>1\} +  P_\ell \;{\mathbbm 1}\{\ell=1\}\Big)  \notag\\
			&\qquad =\Big(1-\frac{M}{N}\Big) P_\ell   .  \label{eq: updated}
			\end{align}
			There are $B$ root nodes in the augmented conflict graph corresponding to each of the requested files that have not been updated. Similar to \cite[Appendix A]{ji2017order} it can be shown that for $B\rightarrow\infty$
			\begin{align}
			\sum\limits_{\substack{ v_r \in \Vc_r: \\ \rho(v_r)\ni \Qbf_k }  }    \frac{ \Ysf(\Kc_{\ell}, \Gc_{v_r})}{ B} \;\;{ \stackrel{p}{\rightarrow }} \;\; 
			\Big(1-\frac{M}{N}\Big) P_\ell\label{eq:updated}
			\end{align}
		\end{itemize}
		
		By comparing \eqref{eq:not updated} and \eqref{eq:updated} it can be seen that
		\begin{align}
		\max\limits_{k\in \Kc_{\ell}}    \sum\limits_{\substack{ v_r \in \Vc_r: \\ \rho(v_r) \ni \Qbf_k  }  }       
		\frac{\Ysf(\Kc_{\ell}, \Gc_{v_r})}{B}   \;\;{ \stackrel{p}{\rightarrow }} \;\; \Big(1-\frac{M}{N}\Big) P_\ell
		\end{align}
		Then from \eqref{eq: colors}, it follows that
		\begin{align}
		\frac{1}{B} \EE_{\Cbf}\bigg[ \EE_{\Qbf}\Big[   \mathcal J(\Cbf,\Qbf )  |    \Cbf   \Big]  \bigg]  & =   
		\frac{1}{B} \sum\limits_{\ell =1 }^{K} \sum\limits_{\Kc_{\ell} \subseteq \Kc }   \EE_{\Cbf}\Bigg[   \EE_{\Qbf} \bigg[   \max\limits_{k\in \Kc_{\ell}}    \sum\limits_{\substack{ v_r \in \Vc_r: \\ \rho(v_r) \ni \Qbf_k  }  } \hspace{-2mm}      \Ysf(\Kc_{\ell}, \Gc_{v_r})    \Big|   \Cbf \bigg]   \Bigg],   \notag\\
		& =  \sum\limits_{\ell =1 }^{K} \binom{K}{\ell} \Big(1-\frac{M}{N}\Big) P_\ell . \notag 
		\end{align}

	\end{itemize}
	
	\subsection{Refinement Segment Rate Achieved with GGC$_1$:}
	We upper bound the refinement rate needed for lossless reconstruction, by the number of distinct files requested by the $\widetilde K$ receivers from the  $\widetilde N$ updated files. Therefore, 
	\begin{align}
	\EE \bigg[\EE_{\dbf}\Big[ &\text{Number of distinct requests from the updated library}  \Big]  \bigg] \notag\\
	&  \qquad\qquad = \EE \bigg[ \sum\limits_{n = 1}^{ \widetilde N}  \EE_{\dbf}\Big[ \mathbbm 1\Big\{ \text{Updated file $\File_n^F$ is requested}\Big\}\Big]  \bigg]  \notag \\
	&  \qquad\qquad =  \EE \Big[\widetilde N\Big(   1- \Big(1-\frac{1}{\widetilde N }\Big)^{\widetilde  K}     \Big)\Big]  \notag\\
	& \qquad\qquad  \stackrel{(a)}{\leq}\EE \bigg[ \EE \Big[   \widetilde N\Big(   1- \Big(1-\frac{1}{\widetilde N }\Big)^{\widetilde  K}     \Big) \Big|\widetilde N\Big] \bigg]  \notag\\
	& \qquad\qquad \stackrel{(b)}{\leq} \EE \Big[     \widetilde N\Big(   1- \Big(1-\frac{1}{\widetilde N }\Big)^{   K_\pi}     \Big)   \Big]  \notag\\
	& \qquad\qquad \stackrel{(c)}{\leq}    N_{\pi}\Big(   1- \Big(1-\frac{1}{N_{\pi}}\Big)^{   K_\pi}     \Big),   \notag\\
	\end{align}
	where (a) follows by applying the law of total expectation, and (b) and (c) follow by applying Jensen's inequality, and since $\EE[\widetilde K]= \pi K\triangleq K_\pi$ and $\EE[\widetilde N]= \pi N\triangleq N_\pi$.
	
	\vspace{2mm}
	The effective group coloring done by Algorithm GGC$_1$ is described as follows. Even though the updated version of a file is not locally available at any of the receivers, since it is correlated with the original version that was prefetched, if an updated file is requested, the sender is able to explore coding opportunities by: $i)$ using the cached packets to deliver the original version of the file, and $ii)$ then sending refinements that enable the receiver to recover the most recent version of the requested file. On the other hand, a strategy that is agnostic to the correlation among the different versions of a file would deliver the requested content through uncoded transmissions.

	\subsection{Rate Achieved with GGC$_2$}
	As in Appendix~\ref{app:uniform rate}-B, we upper bound the rate achieved with Algorithm GGC$_2$ by computing the number of distinct requested files, such that none of the files belong to the file-level $\delta$-ensemble of another file. Given that the sender always delivers the demand from the updated library, then, only $\File_{n}^F$ may be requested from the file-level $\delta$-ensemble $E_n=\{\File_n^F, W_n^F \}$. Therefore, the multicast codeword is composed of a concatenation of {\em all} the  requested files. As in \eqref{eq:GCC2 1}, an upper bound on the rate is given as
	\begin{align}
	\EE\Big[ \text{Number of selected distinct requested files} \Big]  +&
	\delta\; \EE\Big[ \text{Number of distinct requests not selected} \Big]\notag \\
	&\leq N\Big(   1- \Big(1-\frac{1}{N}\Big)^{K}     \Big).
	\end{align}

	\section{Proof of Theorem \ref{thm:optimality of scheme}}\label{app:optimality of scheme}
	A lower bound on the optimal average rate-memory function for a library with heterogeneous sources is given in \cite{ISITjournal}. For the setting considered in Sec.~\ref{Sec: 2user2file}, $H(W_1)=H(W_2)=H(W)$,   $H(W_1,W_2)=(1+\delta)H(W)$, and $I(W_1;W_2)=(1-\delta)H(W)$. Therefore, $R^*(M)$ is lower bounded as
	\begin{itemize}
		\item[-] $M \in  \Big[0,  \,  H(W) \Big),$ 
		$$ R^*(M)\geq   \frac{1}{2}  \Big(H(W_1,W_2)+ H(W)\Big) - M = (1+\frac{\delta}{2})H(W)  - M  ,$$
		\item[-] $M \in \Big[  H(W), \, H(W_1,W_2)\Big),$ 
		$$R^*(M)\geq     \frac{1}{2}\Big( H(W_1,W_2)- M\Big)=\frac{1}{2}\Big((1+\delta)H(W)  - M\Big) .\qquad\qquad$$
		\item[-]	$M \in \Big[   H(W_1,W_2),H(W_1)+H(W_2)\Big],$ 
		$$R^*(M)=0 .\qquad\qquad$$
	\end{itemize}
	By comparing the achievable rate $\Rach(M)$, given in \eqref{eq: rate 2 user 2file}, with the lower bound we have the following.
	\begin{itemize}
		\item When $M \in  \Big[0,  \,  H(W)\Big)$,
		\begin{align} 
		\Rach(M)  - R^*(M) &\leq  \Big(1+\frac{\delta}{2}\Big) \Big(H(W)-M\Big) + \min\Big\{\frac{1}{2}, {\delta}\Big\}  M   -  \Big((1+\frac{\delta}{2})H(W)  - M  \Big)  \notag\\
		& = \frac{1}{2} \min\Big\{ {1-\delta} , {\delta}  \Big\}M   \notag\\
		&  \stackrel{(a)}{\leq} \frac{1}{2} \min\Big\{ {1-\delta} , {\delta} \Big\}  H(W)  \notag\\
		&=  \frac{1}{2}\min\Big\{  I(W_1;W_2)  , H(W_1|W_2) \Big\}  ,\notag 
		\end{align}
		where in (a) the gap is upper bounded using $M< H(W)$.
		
		\item When $M \in  \Big[ H(W), H(W_1,W_2)\Big)$,
		\begin{align} 
		\Rach(M)  - R^*(M) &\leq   \min\Big\{\frac{1}{2},{\delta} \Big\}   \Big(2H(W)-M\Big)     -  \frac{1}{2}\Big( (1+\delta)H(W)-M\Big) \notag\\
		& =  \min\Big\{\frac{1-\delta}{2},\frac{3\delta-1}{2}  \Big\} H(W)+ \min\Big\{0, \frac{1-\delta}{2}  M  \Big\} \notag\\ 
		& = \min\Big\{\frac{1-\delta}{2},\frac{3\delta-1}{2}  \Big\}  H(W)\notag\\
		& \leq  \frac{1-\delta}{2 }  H(W) = \frac{1}{2} I(W_1;W_2).\notag
		\end{align}
		
		\item When $M \in  \Big[ H(W_1,W_2) ,2H(W)\Big]$, 
		\begin{align} 
		\Rach(M)  - R^*(M) &\leq  \min\Big\{\frac{1}{2},{\delta} \Big\}   \Big(2H(W)-M\Big)       \notag\\
		& \stackrel{(b)}{\leq}  \min\Big\{\frac{1}{2}, {\delta}  \Big\}(1-\delta)   H(W) \notag\\
		& \leq  \frac{1-\delta}{2 } H(W)= \frac{1}{2} I(W_1;W_2),\notag
		\end{align}
		where in (b) the gap is upper bounded using $M\geq H(W_1,W_2) = (1+\delta)H(W)$.
	\end{itemize}

\end{appendices}

\bibliographystyle{IEEEtran}
\bibliography{References}

\begin{thebibliography}{10}
\providecommand{\url}[1]{#1}
\csname url@samestyle\endcsname
\providecommand{\newblock}{\relax}
\providecommand{\bibinfo}[2]{#2}
\providecommand{\BIBentrySTDinterwordspacing}{\spaceskip=0pt\relax}
\providecommand{\BIBentryALTinterwordstretchfactor}{4}
\providecommand{\BIBentryALTinterwordspacing}{\spaceskip=\fontdimen2\font plus
\BIBentryALTinterwordstretchfactor\fontdimen3\font minus
  \fontdimen4\font\relax}
\providecommand{\BIBforeignlanguage}[2]{{%
\expandafter\ifx\csname l@#1\endcsname\relax
\typeout{** WARNING: IEEEtran.bst: No hyphenation pattern has been}%
\typeout{** loaded for the language `#1'. Using the pattern for}%
\typeout{** the default language instead.}%
\else
\language=\csname l@#1\endcsname
\fi
#2}}
\providecommand{\BIBdecl}{\relax}
\BIBdecl

\bibitem{maddah14fundamental}
M.~A. Maddah-Ali and U.~Niesen, ``Fundamental limits of caching,'' \emph{IEEE
  Transactions on Information Theory}, vol.~60, no.~5, pp. 2856--2867, 2014.

\bibitem{maddah14decentralized}
M.~Maddah-Ali and U.~Niesen, ``Decentralized coded caching attains
  order-optimal memory-rate tradeoff,'' \emph{IEEE/ACM Transactions on
  Networking}, no.~99, pp. 1--8, 2014.

\bibitem{ji2017order}
M.~Ji, A.~M. Tulino, J.~Llorca, and G.~Caire, ``Order-optimal rate of caching
  and coded multicasting with random demands,'' \emph{IEEE Transactions on
  Information Theory}, vol.~63, no.~6, pp. 3923--3949, June 2017.

\bibitem{ji14average}
------, ``On the average performance of caching and coded multicasting with
  random demands,'' in \emph{Proc. IEEE International Symposium on Wireless
  Communications Systems (ISWCS)}, 2014, pp. 922--926.

\bibitem{ji14groupcast}
------, ``Caching and coded multicasting: Multiple groupcast index coding,'' in
  \emph{in Proc. IEEE Global Conference on Signal and Information Processing
  (GlobalSIP)}, 2014, pp. 881--885.

\bibitem{ji15multiple}
------, ``Caching-aided coded multicasting with multiple random requests,'' in
  \emph{in Proc. IEEE Information Theory Workshop (ITW)}, 2015, pp. 1--5.

\bibitem{channel2016}
A.~Cacciapuoti, M.~Caleffi, M.~Ji, L.~J., and A.~Tulino, ``Speeding up future
  video distribution via channel-aware caching-aided coded multicast,''
  \emph{IEEE JSAC}, vol.~34, no.~8, pp. 2207--2218, 2016.

\bibitem{shanmugam14finite}
K.~Shanmugam, M.~Ji, A.~Tulino, J.~Llorca, and A.~Dimakis, ``Finite length
  analysis of caching-aided coded multicasting,'' in \emph{in Proc. IEEE Annual
  Allerton Conference on Communication, Control, and Computing}, Oct. 2014.

\bibitem{yu2016exact}
Q.~Yu, M.~A. Maddah-Ali, and A.~S. Avestimehr, ``The exact rate-memory tradeoff
  for caching with uncoded prefetching,'' \emph{IEEE Transactions on
  Information Theory}, vol.~64, no.~2, pp. 1281--1296, 2018.

\bibitem{iot2012}
D.~Miorandi, S.~Sicari, F.~D. Pellegrini, and I.~Chlamtac, ``Internet of
  things: Vision, applications and research challenges,'' \emph{Ad Hoc
  Networks}, vol.~10, no.~7, pp. 1497 -- 1516, 2012.

\bibitem{Antony17aoi}
C.~Kam, S.~Kompella, J.~E. Nguyen, G. D. and.~Wieselthier, and E.~A.,
  ``Information freshness and popularity in mobile caching,'' \emph{Proc. IEEE
  International Symposium on Information Theory (ISIT)}, 2017.

\bibitem{birk1998informed}
Y.~Birk and T.~Kol, ``Informed-source coding-on-demand (iscod) over broadcast
  channels,'' in \emph{Proc. IEEE INFOCOM'98. Seventeenth Annual Joint
  Conference of the IEEE Computer and Communications Societies.}, vol.~3.\hskip
  1em plus 0.5em minus 0.4em\relax IEEE, 1998, pp. 1257--1264.

\bibitem{IndexCoding}
Z.~Bar-Yossef, Y.~Birk, T.~Jayram, and T.~Kol, ``Index coding with side
  information,'' \emph{IEEE Transactions on Information Theory}, vol.~57,
  no.~3, pp. 1479--1494, 2011.

\bibitem{timo2018rate}
R.~Timo, S.~S. Bidokhti, M.~Wigger, and B.~C. Geiger, ``A rate-distortion
  approach to caching,'' \emph{IEEE Transactions on Information Theory},
  vol.~64, no.~3, pp. 1957--1976, 2018.

\bibitem{hassanzadeh2017rate}
P.~Hassanzadeh, A.~Tulino, J.~Llorca, and E.~Erkip, ``Rate-memory trade-off for
  the two-user broadcast caching network with correlated sources,'' \emph{in
  Proc. IEEE International Symposium on Information Theory (ISIT)}, 2017.

\bibitem{Asilomar2017}
------, ``Broadcast caching networks with two receivers and multiple correlated
  sources,'' \emph{in Proc. IEEE Asilomar Conference on Signals, Systems, and
  Computers}, 2017.

\bibitem{ISITjournal}
------, ``Rate-memory trade-off for the broadcast caching network with
  correlated sources,'' \emph{Available on arXiv.}

\bibitem{ITW2016}
------, ``Correlation-aware distributed caching and coded delivery,'' \emph{in
  Proc. IEEE Information Theory Workshop (ITW)}, 2016.

\bibitem{yang2017centralized}
Q.~Yang and D.~G{\"u}nd{\"u}z, ``Centralized coded caching of correlated
  contents,'' \emph{arXiv preprint arXiv:1711.03798}, 2017.

\bibitem{ISTC2016}
P.~Hassanzadeh, A.~Tulino, J.~Llorca, and E.~Erkip, ``Cache-aided coded
  multicast for correlated sources,'' \emph{in Proc. IEEE International
  Symposium on Turbo Codes and Iterative Information Processing (ISTC)}, 2016.

\bibitem{gray1974source}
R.~Gray and A.~Wyner, ``Source coding for a simple network,'' \emph{Bell System
  Technical Journal}, vol.~53, no.~9, pp. 1681--1721, 1974.

\end{thebibliography}

\end{document}